\documentclass[11pt]{article}

\newcommand{\ifabs}[2]{#2}

\ifabs{\usepackage[margin = 1in]{geometry}}{
\usepackage{geometry}
}

\usepackage{amsmath}
\usepackage{amsfonts}
\usepackage{amsthm}
\usepackage{amssymb}
\usepackage{color}
\usepackage{caption}
\ifabs{\usepackage{times}}{}

\usepackage{hyperref, color}
\hypersetup{colorlinks=true,citecolor=blue, linkcolor=blue, urlcolor=blue}

\usepackage{comment}
\excludecomment{suppress}

\newtheorem{theorem}{Theorem}[section]

\newtheorem{claim}[theorem]{Claim}
\newtheorem*{claim*}{Claim}

\newtheorem{lemma}[theorem]{Lemma}
\newtheorem{proposition}[theorem]{Proposition}
\newtheorem{corollary}[theorem]{Corollary}
\theoremstyle{definition}

\newtheorem{definition}{Definition}[section]
\newtheorem{remark}{Remark}[section]
\newtheorem*{remark*}{Remark}

\newcommand{\concept}[1]{\emph{{#1}}}

\newcommand{\todo}[1]{\typeout{TODO: \the\inputlineno: #1}\textbf{\color{red}[[[ #1 ]]]}}

\newcommand{\LOCAL}{$\mathsf{LOCAL}$}
\newcommand{\SLOCAL}{$\mathsf{SLOCAL}$}
\newcommand{\DTV}[2]{d_{\mathrm{TV}}\left({#1},{#2}\right)}
\newcommand{\err}[2]{\mathsf{err}\left({#1},{#2}\right)}
\newcommand{\JVV}{local-JVV}

\newcommand{\poly}{\mathrm{poly}}
\newcommand{\Problem}[1]{\mathfrak{{#1}}}

\newcommand{\Ball}{{B}}

\newcommand{\dist}{\mathrm{dist}}
\renewcommand{\vec}[1]{\boldsymbol{{#1}}}

\newcommand{\alg}[1]{\hat{{#1}}}

\title{On Local Distributed Sampling and Counting\footnote{This research is supported by the National Science Foundation of China under Grant No.~61672275 and No.~61722207.}
}
\author{
Weiming Feng~\thanks{Department of Computer Science and Technology, Nanjing University. Email: {fengwm@smail.nju.edu.cn}.}
\and
Yitong Yin\thanks{State Key Laboratory for Novel Software Technology, Nanjing University. Email: {yinyt@nju.edu.cn}.}
}
\date{}

\begin{document}
\maketitle

\begin{abstract}
In classic distributed graph problems, each instance on a graph specifies a space of feasible solutions (e.g.~all proper ($\Delta+1$)-list-colorings of the graph), and the task of distributed algorithm is to construct a feasible solution using local information. 

We study distributed sampling and counting problems, in which each instance specifies a joint distribution of feasible solutions. The task of distributed algorithm is to sample from this joint distribution, or to locally measure the volume of the probability space via the marginal probabilities. The latter task is also known as inference, which is a local counterpart of counting.

For self-reducible classes of instances, the following equivalences are established  in the \LOCAL{} model up to polylogarithmic factors:
\begin{itemize}
\item For all joint distributions, approximate inference and approximate sampling are computationally equivalent. 
\item For all joint distributions defined by local constraints, 
exact sampling is reducible to either one of the above tasks.
\item 
If further, sequentially constructing a feasible solution is trivial locally,
then all above tasks are easy if and only if the joint distribution exhibits strong spatial mixing. 
\end{itemize}
Combining with the state of the arts of strong spatial mixing, we obtain efficient sampling algorithms in the \LOCAL{} model for various important sampling problems,
including: an $O(\sqrt{\Delta}\log^3n)$-round algorithm for exact sampling matchings in graphs with maximum degree $\Delta$, and an $O(\log^3n)$-round algorithm for sampling according to the hardcore model (weighted independent sets) in the uniqueness regime, which along with the $\Omega(\mathrm{diam})$ lower bound in~\cite{feng2017sampling} for sampling according to the hardcore model in the non-uniqueness regime, gives the first computational phase transition for distributed sampling.
\end{abstract}

\setcounter{page}{0} \thispagestyle{empty} \vfill
\pagebreak

\section{Introduction}
In local computation, classic distributed graph problems are formulated in such a way that each instance $\mathcal{I}$ on a graph $G=(V,E)$ specifies a set $\Omega_{\mathcal{I}}$ of feasible solutions $\vec{y}=(y_v)_{v\in V}$, and the goal of the distributed algorithm is to {construct} a feasible solution $\vec{y}\in\Omega_{\mathcal{I}}$ by outputting $y_v$ at each node $v\in V$. 
As a paradigm, we consider the list-coloring problem, where each instance $\mathcal{I}$ gives a graph $G=(V,E)$ with each node $v\in V$ associated with a list $L_v$ of available colors. Then the set $\Omega_{\mathcal{I}}$ contains all proper list-colorings $\vec{y}$ with $y_v\in L_v$ for every $v\in V$, satisfying $y_u\neq y_v$ for all edges $uv\in E$.

Alternatively, we can consider the set of feasible solutions $\Omega_{\mathcal{I}}$ as the sample space of a probability distribution and imagine that each instance $\mathcal{I}$ specifies a joint distribution $\mu_{\mathcal{I}}$ of feasible solutions $\vec{y}=(y_v)_{v\in V}\in\Omega_{\mathcal{I}}$. 
A distributed graph problem is then given by a class of joint distributions $\mu_{\mathcal{I}}$ indexed by instances $\mathcal{I}$. (In the paradigm of list-coloring problem, we may take each $\mu_{\mathcal{I}}$ as the uniform distribution over all proper list-colorings of instance $\mathcal{I}$.)
To each class of joint distributions of this form, there correspond a number of naturally defined problems.
\begin{itemize}
\item Construction: 
Exhibit a feasible solution $\vec{y}=(y_v)_{v\in V}$ satisfying $\mu_{\mathcal{I}}(\vec{y})>0$, where each node $v\in V$ outputs $y_v$. This is the task for classic distributed graph problems. 
(e.g.~Construct a proper list-coloring of instance $\mathcal{I}$.)
\item Sampling: 
Generate a random solution $\vec{Y}=(Y_v)_{v\in V}$ distributed according to $\mu_\mathcal{I}$, where each node $v\in V$ outputs $Y_v$.
(e.g.~Generate uniformly a proper list-coloring of instance $\mathcal{I}$  at random.)
\item Inference (Counting): 
Each node calculates the marginal probabilities of the random variable $Y_v$ being specific values, where the random vector $\vec{Y}=(Y_v)_{v\in V}$ is drawn from $\mu_\mathcal{I}$.
(e.g.~Each node $v\in V$ calculates the probabilities of $v$ being assigned specific  colors in a uniform random proper list-coloring of instance $\mathcal{I}$.)
\end{itemize}

We choose the inference problem as a local version of counting, as the marginal probabilities are typical local knowledges regarding the volume of probability space.
A more standard global definition of counting is to estimate the number of feasible solutions $|\Omega_{\mathcal{I}}|$ (or the total weights if the joint distribution $\mu_{\mathcal{I}}$ is non-uniform), which is unsuitable to study for local computation because it computes a global information. 
However, it is well known that for self-reducible problems, such global information can be decomposed via the chain rule into the marginal probabilities computed by  inference problems~\cite{jerrum2003book}.
Furthermore, the inference problem itself is especially well-motivated by distributed machine learning~\cite{paskin2004robust}.




Previous studies in local computation were focused on the complexity of constructing a feasible solution. The studies of sampling problems in local computation were started very recently~\cite{feng2017sampling, feng2018coloring}.
Several fundamental questions regarding the local complexities of sampling and counting need to be answered.

\vspace{6pt}
\noindent\textbf{Question 1:} \emph{What is the relation between sampling and counting in local computation?}

It is well known that for self-reducible problems, approximate counting and approximate sampling are inter-reducible on polynomial-time Turing machines~\cite{JVV86}. A natural question is whether this is true for local computation. 
To see the nontriviality of the question, recall that the generic reduction from sampling to counting has to be sequential because the procedure is fully-adaptive: Each individual variable is sampled according to the marginal distribution conditioning on the outcomes of previous samplings.
To understand the relation between counting and sampling in local computation, one has to answer the following question first:
\ifabs{}{\begin{quote}}``How much non-adaptively can we sample a random vector $\vec{Y}=(Y_i)$ by accessing to marginal distributions of individual variables $Y_i$?'' 
\ifabs{}{\end{quote}}which by itself is a fundamental question with a broader background.


\vspace{6pt}
\noindent\textbf{Question 2:} \emph{What are the roles of approximation in  sampling and counting in local computation?}

Due to the pioneering works of Valiant~\cite{valiant1979complexity} and Stockmeyer~\cite{stockmeyer1983complexity}, on Turing machines exact counting can be much harder then approximate counting. 
Similar phenomena occur for local computation: Due to a straightforward information-theoretical argument, exact inference is impossible to compute locally unless for joint distributions with zero long-range correlation.

For sampling, a celebrated result of Jerrum, Valiant, and Vazirani~\cite{JVV86} shows that on polynomial-time Turing machines, for self-reducible problems, approximate sampling can be boosted into exact sampling via a rejection sampling procedure (the JVV sampler), such that the algorithm succeeds with high probability (with certifiable failures) and conditioning on success the output is distributed precisely according to the joint distribution $\mu_{\mathcal{I}}$ specified by the instance $\mathcal{I}$. Our question is whether there is a distributed variant of the JVV sampler. 
Answering this question also involves investigating a fundamental problem:
``Can rejection sampling be made local?''
Considering that the rejection sampling as a basic Monte Carlo method has been studied for more than half a century, this is certainly worth studying in local computation.

\vspace{6pt}
\noindent\textbf{Question 3:} \emph{What makes a sampling or counting problem solvable by local computations?}

Finally, we want to characterize the easiness of sampling and counting in local computation by the properties of joint distributions. For sampling and counting on  polynomial-time Turing machines, a phase transition of computational complexity is witnessed at the threshold for the decay of correlation (strong spatial mixing)~\cite{weitz2006counting, sly2010computational}. We wonder whether similar computational phase transitions exist for local computation.
 

\subsection{Our results}
We study the local complexities of self-reducible sampling and counting problems, and provide answers to above fundamental questions.

We formulate distributed sampling and inference (counting) problems  by classes of joint distributions of solutions. Besides general joint distributions, we focus on the joint distributions with following properties naturally arising from local computation:
\begin{itemize}
\item[($\star$)] Joint distributions defined by local constraints. These are the counterparts of locally checkable labelings (LCL) in the world of sampling and counting.
In our paper, these joint distributions are called local Gibbs distributions \ifabs{(Definition~\ref{def:Gibbs}).}{(Definition~\ref{def:Gibbs} and~\ref{def:local-Gibbs}).}
\item[($\star$$\star$)] Joint distributions for which constructing a feasible solution is trivial for a sequential local oblivious procedure. For example, this includes distributions over $(\Delta+1)$-list-colorings, but not the distributions over $\Delta$-list-colorings. The property is related to the ergodicity of the local dynamics on feasible solutions.
In our paper, joint distributions of this property are called locally admissible (Definition~\ref{def:local-admissible}).
\end{itemize}

\ifabs{\noindent\textbf{Main results:}}{
\paragraph{Main results:}} For self-reducible classes of instances, we show the followings hold for sampling and inference (counting) in the \LOCAL{} model:
\begin{itemize}
\item 
Approximate inference and approximate sampling are inter-reducible, 
in a sense that 
if one of the tasks is tractable in the \LOCAL{} model so is the other one (Theorem~\ref{thm:approx-infer-approx-sample} and Theorem~\ref{thm:approx-sample-approx-infer}),
where an approximate problem is tractable in the \LOCAL{} model if the problem is solvable for any $n$ and $\delta>0$, where $n$ is the number of nodes and $\delta$ is the approximation error, within time complexity $\mathrm{poly}(\log n,\log\frac{1}{\delta})$.

\item 
If property ($\star$) is satisfied, then
approximate inference, approximate sampling, and exact sampling are all inter-reducible. 
With the above results, this is proved by a distributed JVV sampler that successfully terminates within $\mathrm{polylog}(n)$ rounds with high probability, and conditioning on success, returns a solution distributed precisely as the desired distribution (Theorem~\ref{thm:approx-infer-exact-sample}).

\item 
If further, property ($\star$$\star$) is satisfied, then
all above tasks are easy if and only if the joint distribution exhibits certain degrees of strong spatial mixing (Theorem~\ref{thm:ssm-approx-infer} and Corollary~\ref{coro:ssm-sampling}), a decay of correlation property which is critically related to the computational complexity of approximate counting~\cite{weitz2006counting,sly2010computational}.
\end{itemize}

\ifabs{\noindent\textbf{Minor results:}}{
\paragraph{Minor results:}} As by-products of above  results, we obtain two boosting lemmas:
\begin{itemize}
\item 
If property ($\star$) is satisfied, then
approximate inference with bounded additive (total variation) error can be boosted to  with bounded multiplicative error (Lemma~\ref{lemma-local-boosting}).

\item 
If further, property ($\star$$\star$) is satisfied, then
strong spatial mixing in total variation distance implies strong spatial mixing with decay in multiplicative error (Corollary~\ref{corollary-ssm-boosting}). 

This explains why so far the known strong spatial mixing results for several major problems (e.g.~independent sets, matchings, and graph colorings) were all proved in this stronger form with decay in multiplicative error~\cite{weitz2006counting, bayati2007simple, gamarnik2007correlation, gamarnik2013strong}.
\end{itemize}
Interestingly, the second boosting result states a proposition in probability theory, which seems unrelated to distributed algorithms, but is proved by us via local computation.

\ifabs{\vspace{6pt}\noindent\textbf{Applications in the \LOCAL{} model:}}{
\paragraph{Applications in the \LOCAL{} model:}} The  main results in above together with the state of the arts of strong spatial mixing~\cite{bayati2007simple, weitz2006counting, gamarnik2013strong, li2013correlation, song2016counting}, imply the following exact sampling algorithms for various important sampling problems and statistical physics models:
\begin{itemize}
\item An $O(\sqrt{\Delta}\log^3n)$-round algorithm for exact sampling matchings in graphs with maximum degree $\Delta$.
\item An $O(\log^3n)$-round exact sampling according to the hardcore model (weighted independent sets) in the uniqueness regime, which along with the $\Omega(\mathrm{diam})$ lower bound in a previous work~\cite{feng2017sampling} for sampling according to the hardcore model in the non-uniqueness regime,  gives the first  \concept{computational phase transition} for distributed sampling  and counting.
\item $O(\log^3n)$-round exact sampling algorithms for sampling according to various anti-ferromagnetic models, including: anti-ferromagnetic 2-spin model in the uniqueness regime, weighted hypergraph matchings in the uniqueness regime, proper $q$-colorings of triangle-free graphs when $q\ge\alpha\Delta$ where $\alpha>\alpha^*$ and $\alpha^*\approx1.763$ is the positive root of the euqation $x=\mathrm{e}^{1/x}$.
\end{itemize}

\ifabs{}{
\subsection{Organization of the paper}
Models and definitions are introduced in Section~\ref{sec:model}.
In Section~\ref{sec:SLOCAL}, we apply a generic transformation from local sequential sampling to local distributed sampling, which proves the first main result. 
In Section~\ref{sec:JVV}, we develop the techniques of local self-reductions, which proves the second main result, giving the distributed JVV sampler, along with the boosting lemma for approximate inference.
In Section~\ref{sec:SSM}, we explore the intrinsic relation between decay of correlations and distributed counting and sampling, proving the third main result.
}



\section{Models and Definitions}\label{sec:model}
\ifabs{
\noindent\textbf{The \LOCAL{} Model:}
In the \LOCAL{} model~\cite{naor1995can, peleg2000distributed}, the network is a simple, undirected graph $G=(V,E)$. Initially, each node $v\in V$ receives a local input and an arbitrary long random bit string sampled independently at $v$. For a \LOCAL{} algorithm with time complexity $t$, each node $v\in V$ gathers all information within radius $t$ from $v$, including the topology of the graph, the inputs and random bits of the nodes within that radius, and performs an arbitrary local computation with the information to compute an output.
}{
\paragraph{Notation for Graphs:}
Let $G=(V,E)$ be a simple, undirected graph.
For any two vertices $u,v\in V$, the distance between $u$ and $v$ in $G$, denoted by $\dist_G(u,v)$, is the length of the shortest path between $u$ and $v$ in $G$; and for a subset $S\subseteq V$, we define $\dist_G(v,S)\triangleq\min_{u\in S}\dist_G(u,v)$. For any vertex $v\in V$ and $r>0$, let $\Ball_r(v)\triangleq\{u\in V\mid \dist_G(u,v)\le r\}$ denote the \concept{$r$-ball} centered at $v$ in $G$.
}

\ifabs{\vspace{6pt}
\noindent\textbf{Joint Distributions:}}{
\paragraph{Notation for Joint Distributions:}}
Let $V$ be a set of size $n=|V|$ and $\Sigma$ an \concept{alphabet} of size $q=|\Sigma|\le \poly(n)$. Let $\Omega=\Sigma^V$ be the \concept{sample space}. 
Each $\sigma\in\Sigma^V$ is called a \concept{configuration}. 
For $S\subseteq V$, we use $\sigma_S$ or $\sigma(S)$ to denote the restriction of $\sigma$ on subset $S$.

Let $\mu$ be a {distribution} over $\Sigma^V$, called a \concept{joint distribution}, because each $Y\sim \mu$ is a random vector ${Y}=(Y_v)_{v\in V}$ consisting of $n$ jointly distributed random variables.
For $R\subseteq V$, let $\mu_R$ denote the \concept{marginal distribution} over $\Sigma^R$ induced by $\mu$ on subset $R$\ifabs{.}{, formally:
\[
\forall \tau\in\Sigma^R:\quad \mu_R(\tau)=\Pr_{{Y}\sim\mu}[{Y}_R=\tau]=\sum_{\substack{\sigma\in\Sigma^V: \sigma_R=\tau}}\mu(\sigma).
\]}
In particular, when $R=\{v\}$ for some $v\in V$, we write $\mu_v=\mu_{\{v\}}$. 
A configuration $\sigma\in\Sigma^V$ is \concept{feasible} with respect to $\mu$ if $\mu(\sigma)>0$.
\ifabs{A configuration $\tau\in\Sigma^\Lambda$ on a subset $\Lambda\subseteq V$ is feasible with respect to $\mu$ if there is a $\sigma\in\Sigma^V$  such that $\mu(\sigma)>0$ and $\sigma_\Lambda=\tau$.
}{A configuration $\tau\in\Sigma^\Lambda$ on a subset $\Lambda\subseteq V$ is feasible with respect to $\mu$ if $\mu_\Lambda(\tau)>0$, i.e.~if there is a feasible $\sigma\in\Sigma^V$  such that $\sigma_\Lambda=\tau$. By convention an empty configuration $\varnothing$ is always feasible.}
Given a feasible configuration $\tau\in\Sigma^\Lambda$ on a subset $\Lambda\subseteq V$, we use $\mu^{\tau}$ to denote the distribution over $\Sigma^V$ induced by $\mu$ conditioning on $\tau$\ifabs{.}{, formally:
\[
\forall \sigma\in\Sigma^V:\quad \mu^{\tau}(\sigma)=\Pr_{{Y}\sim \mu}[{Y}=\sigma\mid {Y}_\Lambda=\tau].
\]
}
The conditional marginal distributions $\mu_R^\tau$ and $\mu_v^\tau$ are accordingly defined.


Suppose that $\mu$ and $\nu$ are two distributions over the same sample space $\Omega$. The \concept{total variation distance} between $\mu$ and $\nu$, denoted by $\DTV{\mu}{\nu}$, is defined as:
\ifabs{$\DTV{\mu}{\nu}\triangleq\frac{1}{2}\left\|\mu-\nu\right\|_1=\max_{A\subseteq\Omega}|\mu(A)-\nu(A)|$.}{
\[
\DTV{\mu}{\nu}\triangleq\frac{1}{2}\left\|\mu-\nu\right\|_1=\max_{A\subseteq\Omega}|\mu(A)-\nu(A)|.
\]
}

\ifabs{\vspace{6pt}
\noindent\textbf{Distributed Graph Problems:}}{
\paragraph{Distributed Graph Problems:}}
We reformulate the notion of \concept{distributed graph problems} in~\cite{ghaffari2016complexity} by classes of joint distributions. 

\begin{definition}[distributed graph problems]\label{def:distributed-graph-problem}
A distributed graph problem is given by a class of joint distributions $\Problem{M}=\{\mu_{(G,\vec{x})}\}$, indexed by labeled graphs $(G,\vec{x})$, where $G=(V,E)$ is a simple, undirected graph and $\vec{x}=(x_v)_{v\in V}$ is a $|V|$-dimensional vector. 
Each $\mu_{(G,\vec{x})}$ is a joint distribution over $\Sigma^V$, where $\Sigma=\Sigma_{(G,\vec{x})}$ is an {alphabet} of size \ifabs{$q=|\Sigma|\le\poly(|V|)$.}{$q=|\Sigma|$ bounded in polynomial of $n=|V|$.}

The class of distributions $\Problem{M}$ is \concept{translation-invariant}, which means that if the labeled graphs $(G,\vec{x})$ and $(\tilde{G},\tilde{\vec{x}})$ are isomorphic under bijection $\phi$ on vertices, then $\mu_{(G,\vec{x})}\in \Problem{M}$ is well-defined if and only if $\mu_{(\tilde{G},\tilde{\vec{x}})}$ is, and the distributions $\mu_{(G,\vec{x})}$ and $\mu_{(\tilde{G},\tilde{\vec{x}})}$ are identical under bijection $\phi$.
\end{definition}

Initially, each node $v\in V$ knows $x_v$. For classic distributed graph problems, the nodes need to construct a feasible $\vec{y}$ that  $\mu_{(G,\vec{x})}(\vec{y})>0$, where each node $v\in V$ outputs $y_v$.

\begin{remark}
In~\cite{ghaffari2016complexity}, a distributed graph problem is defined by a relation $\mathcal{T}$ that contains all satisfying tuples $(G,\vec{x},\vec{y})$ instead of as a class of distributions $\Problem{M}=\{\mu_{(G,\vec{x})}\}$ over $\vec{y}$. The two definitions are equivalent: Given a relation $\mathcal{T}$, for any $(G,\vec{x})$, $\mu_{(G,\vec{x})}$ is just a positive distribution over all feasible solutions $\vec{y}$ satisfying $(G,\vec{x},\vec{y})\in\mathcal{T}$ (e.g.~the uniform distribution over feasible solutions); and conversely, given $\Problem{M}=\{\mu_{(G,\vec{x})}\}$, the relation $\mathcal{T}$ can be constructed as $\mathcal{T}=\{(G,\vec{x},\vec{y})\mid \mu_{(G,\vec{x})}(\vec{y})>0\}$.
\end{remark}

\ifabs{
\noindent\textbf{Distributed Sampling and Counting:}}{
\paragraph{Distributed Sampling and Counting:}}
A distributed sampling or counting problem is also given by a class of joint distributions $\Problem{M}=\{\mu_{(G,\vec{x})}\}$ as defined in Definition~\ref{def:distributed-graph-problem}.

\begin{definition}[instances for distributed sampling/counting]\label{def:self-reducible-instance}
Let $\Problem{M}=\{\mu_{(G,\vec{x})}\}$ be a class of joint distributions as defined in Definition~\ref{def:distributed-graph-problem}.
An {instance} for distributed sampling/counting is a tuple $(G,\vec{x},\tau)$,  where $(G,\vec{x})$ specifies a joint distribution $\mu=\mu_{(G,\vec{x})}$ over $\Sigma^V$, and $\tau\in\Sigma^\Lambda$ is an arbitrary configuration on a subset $\Lambda\subseteq V$ that is feasible with respect to $\mu$. 
We call the distribution $\mu^{\tau}$ the \concept{target distribution}. 

Given an instance $(G,\vec{x},\tau)$ where $\tau$ is specified on the subset $\Lambda\subseteq V$, initially each node $v\in V$ knows $x_v$, and also $\tau_v$ if $v\in\Lambda$. We assume that $x_v$ includes a unique ID for $v$, a global polynomial upper bound of $n=|V|$, and a global  upper bound on  errors if approximation is involved.


\end{definition}


\begin{remark}


The reason to include an arbitrary partially specified configuration $\tau$ into the problem instance, is to explicitly enforce the \concept{self-reducibility}, a property that is essential to problems such as sampling and counting~\cite{JVV86}.\ifabs{}{ In our context, it means that arbitrarily fixing a feasible evaluation $\tau$ of a subset of variables, the conditional distribution $\mu^{\tau}$ forms a valid instance over the remaining free variables.}\footnote{Alternatively, one may enforce the self-reducibility implicitly by assuming it as a property of the class of joint distributions $\Problem{M}=\{\mu_{(G,\vec{x})}\}$: For every $\mu_{(G,\vec{x})}\in\Problem{M}$ where $G=(V,E)$, for any feasible configuration $\tau\in\Sigma^\Lambda$ on a subset $\Lambda\subseteq V$, there exists a $\mu_{(G',\vec{x}')}\in\Problem{M}$ such that $G'=(V,E')$ is a subgraph of $G$ with identical vertex set and  $\mu_{(G,\vec{x})}^{\tau}=\mu_{(G',\vec{x}')}$, where 
the new instance $(G',\vec{x}')$ can be constructed locally from $(G,\vec{x})$, providing $\tau_v$ to each node $v\in\Lambda$. This alternative formulation of self-reducibility is equivalent to the one we used above.
}
\ifabs{}{


}For example, consider $\mu$ as the uniform distribution over all proper $(\Delta+1)$-colorings of $G$. For any proper $(\Delta+1)$-coloring $\tau$ of vertices in a subset $\Lambda\subseteq V$, $\mu^{\tau}$ is the uniform distribution over all proper $(\Delta+1)$-colorings of $G$ consistent with $\tau$. This equivalently specifies a uniform distribution over list-colorings of the subgraph $G[V\setminus \Lambda]$ induced by subset $V\setminus\Lambda$, where each node $v\in V$ holds a list of available colors $L_v=[q]\setminus\{\tau_u\mid uv\in E\}$.
\end{remark}


We assume that the time complexity of a distributed algorithm is fixed.
Upon termination the algorithm either successfully returns or fails.
We assume that the algorithm succeeds with high probability, and all failures are \concept{locally certifiable}. 
Upon termination, each node $v\in V$ besides the regular output explicitly outputs a random bit $F_v$ indicating whether the algorithm fails locally at $v$, 
and it is guaranteed that $\sum_{v\in V}\mathbb{E}[F_v]= O(\frac{1}{n})$. This is a well accepted notion of the Las Vegas algorithms for local computation~\cite{ghaffari2017derandomizing}.

The goal of \concept{distributed sampling} is to draw a random sample according to the target distribution $\mu^{\tau}$ given by the instance $(G,\vec{x},\tau)$. 







\begin{description}
\ifabs{\vspace{-6pt}}{}
\item[\textbf{Exact sampling:}] Given any instance $(G,\vec{x},\tau)$, the nodes need to sample a random $\vec{Y}=(Y_v)_{v\in V}$ upon successful termination, where each node $v\in V$ outputs $Y_v$ or fails, such that conditioning on that no one fails the distribution of $\vec{Y}$ is precisely $\mu^\tau$.

\ifabs{\vspace{-6pt}}{}
\item[\textbf{Approximate sampling:}] 
Given any instance $(G,\vec{x},\tau)$, for any $\delta>0$, the nodes need to output a random $\vec{Y}=(Y_v)_{v\in V}$ upon successful termination, 
such that conditioning on success the distribution $\alg{\mu}$ of $\vec{Y}$  satisfies $\DTV{\alg{\mu}}{\mu^{\tau}}\le\delta$.

\end{description}

\ifabs{\vspace{-6pt}}{}
We use the \concept{distributed inference} to represent counting in distributed settings. Here, the goal is to estimate the marginal distribution $\mu_v^{\tau}$ for each node $v\in V$, where $\mu_v^{\tau}$ is the target distribution given by the instance $(G,\vec{x},\tau)$. Due to an information-theoretical argument, exact inference with local information is impossible for joint distributions with nonzero long-range correlations. Hence we focus on approximate inference.

\begin{description}
\ifabs{\vspace{-6pt}}{}
\item[\textbf{Approximate inference:}] 
Given any instance $(G,\vec{x},\tau)$, for any $\delta>0$, 
each node $v\in V$ needs to output a marginal distribution $\alg{\mu}_v$ over $\Sigma$, which is a vector $\alg{\mu}_v\in[0,1]^\Sigma$ and $\|\alg{\mu}\|_1=1$,
satisfying that $\DTV{\alg{\mu}_v}{\mu_v^{\tau}}\le\delta$.






\end{description}
\ifabs{\vspace{-6pt}}{}
A more accurate approximate inference with bounded multiplicative error is discussed in Section~\ref{sec:boosting}. 



\ifabs{\vspace{6pt}
\noindent\textbf{Joint Distributions Defined by Local Constraints:}}{
\paragraph{Joint Distributions Defined by Local Constraints:}}
We use the Gibbs distributions defined by the \concept{weighted constraint satisfaction problems (CSPs)}, also known as the \concept{factor graphs}~\cite{mezard2009information}, to model joint distributions characterized by local constraints.
\ifabs{\vspace{-3pt}}{}
\ifabs{\begin{definition}[local Gibbs distributions]}{\begin{definition}[Gibbs distributions]}\label{def:Gibbs}
A \concept{Gibbs distribution} $\mu$ is specified by a tuple $(G,\Sigma,\mathcal{F})$, where $G=(V,E)$ is an undirected graph, $\Sigma$ is an alphabet of size $q=|\Sigma|$ bounded in polynomial of $n=|V|$, and $\mathcal{F}$ is a collection of \concept{constraints} (also called \concept{factors}). A constraint $(f,S)\in\mathcal{F}$ consists of a nonnegative-valued function $f:\Sigma^S\to\mathbb{R}_{\ge 0}$ defined on the \concept{scope} $S\subseteq V$. \ifabs{}{A constraint $(f,S)$ is a \concept{soft constraint} if the function $f$ is positive-valued, otherwise it is a \concept{hard constraint}.}

Each configuration $\sigma\in\Sigma^V$ is assigned a weight:
\ifabs{$w(\sigma) = \prod_{(f,S)\in\mathcal{F}}f(\sigma_{S})$.}{
\begin{align}
w(\sigma) = \prod_{(f,S)\in\mathcal{F}}f(\sigma_{S}).\label{eq:configuration-weight}
\end{align}}
The {Gibbs distribution} $\mu$ over all configurations in $\Sigma^V$ is defined proportional to the weights
\ifabs{$\mu(\sigma)= \frac{w(\sigma)}{Z}$,}{
\[
\mu(\sigma)= \frac{w(\sigma)}{Z},
\]}
where the normalizing factor $Z=\sum_{\sigma\in\Sigma^V}w(\sigma)$ is known as the \concept{partition function}, which can be seen as a function of the specification $(G,\Sigma,\mathcal{F})$ of the distribution $\mu$. 

\ifabs{
A Gibbs distribution specified by $(G,\Sigma,\mathcal{F})$ 
is \concept{local} if for any constraint $(f,S)\in\mathcal{F}$, the diameter of the scope $S$ in graph $G$ is bounded as $\max_{u,v\in S}\dist_G(u,v)=O(1)$.
}{}
\end{definition}

\ifabs{}{
In particular, when all constraint functions $f$ are Boolean-valued functions, the distribution $\mu$ is the uniform distribution over all feasible configurations, and $Z$ gives the total number of feasible configurations.


\begin{definition}[locality of Gibbs distributions]\label{def:local-Gibbs}

A class of Gibbs distributions specified by $(G,\Sigma,\mathcal{F})$ 
are \concept{local} if for any constraint $(f,S)\in\mathcal{F}$, the diameter of the scope $S$ in graph $G$ is bounded by a constant, that is, $\max_{u,v\in S}\dist_G(u,v)=O(1)$.
\end{definition}
}



\ifabs{\vspace{-3pt}}{}
The local Gibbs distributions are the counterparts of the LCL problems~\cite{naor1995can} in the world of distributed sampling/counting.
Just as that  the LCL problems are the distributed graph problems that are defined by local constraints, the local Gibbs distributions are the joint distributions  that are defined by local factors. 

\ifabs{}{
When an instance $(G,\vec{x},\tau)$ is provided to distributed algorithm, if the joint distribution $\mu_{(G,\vec{x})}$ is a local Gibbs distribution specified by $(G,\Sigma,\mathcal{F})$, we assume that for each node $v\in V$, $x_v$ includes the descriptions of all local constrains $(f,S)\in\mathcal{F}$ that $v\in S$.

An important property of Gibbs distribution is the \concept{spatial Markovian} property, also known as conditional independence, which is stated formally by the following proposition. 

\begin{proposition}[conditional independence]\label{prop:cond-ind}
Let $\mu$ be a Gibbs distribution specified by $(G,\Sigma,\mathcal{F})$, where $G=(V,E)$. Let $H=(V,F)$ denote the hypergraph with vertices $V$ and hyperedges $F=\{S\mid (f,S)\in\mathcal{F}\}$. 
Suppose that $A,B,C\subset V$ are disjoint nonempty subsets and $C$ is a vertex separator whose removal disconnects $A$ and $B$ in $H$.
For a random vector $Y\in\Sigma^V$ distributed according to $\mu$,  $Y_A$ and $Y_B$ are conditionally independent given that $Y_C$ is arbitrarily and feasibly fixed. Formally, for any $\sigma_C\in\Sigma^C$ that is feasible with respect to $\mu$, any $\sigma_A\in\Sigma^A$ and $\sigma_B\in\Sigma^B$, it holds that
\[
\Pr_{Y\sim\mu}[Y_A=\sigma_A\wedge Y_B=\sigma_B\mid Y_C=\sigma_C]=
\Pr_{Y\sim\mu}[Y_A=\sigma_A \mid Y_C=\sigma_C]\Pr_{Y\sim\mu}[Y_B=\sigma_B \mid Y_C=\sigma_C].
\]
\end{proposition}
}


We also consider a restrictive class of Gibbs distributions, with the following property.

\begin{definition}[locally admissible]\label{def:local-admissible}
Let $\mu$ be a Gibbs distribution specified by $(G,\Sigma,\mathcal{F})$ where $G=(V,E)$. Let $\Lambda\subseteq V$ be a subset. A configuration $\sigma\in\Sigma^\Lambda$ on a subset $\Lambda\subseteq V$ is \concept{locally feasible} if $\sigma$ itself does not violate any constraint onsite, 
\ifabs{i.e.~$\prod_{(f,S)\in\mathcal{F}\atop S\subseteq \Lambda}f(\sigma_S) >0$.}{
formally:
\[
\prod_{(f,S)\in\mathcal{F}\atop S\subseteq \Lambda}f(\sigma_S) >0.
\]}

\ifabs{
A Gibbs distribution $\mu$ is said to be \concept{locally admissible} if $\forall \Lambda\subseteq V$, $\forall \sigma\in\Sigma^\Lambda$, $\sigma$ is feasible $\iff$ $\sigma$ is locally feasible.
}{
Recall that a $\sigma\in\Sigma^\Lambda$ is feasible if there is a $\tau\in\Sigma^V$ such that $\mu(\sigma)>0$ and $\tau_\Lambda=\sigma$. Clearly $\sigma$ is locally feasible if it is feasible. 
A Gibbs distribution $\mu$ is said to be \concept{locally admissible} if the converse is also true:
\[
\forall \Lambda\subseteq V,\,\, \forall \sigma\in\Sigma^\Lambda:
\quad
\sigma\text{ is feasible }\iff \sigma \text{ is locally feasible.}
\]}
\end{definition}

\begin{remark}
The locally admissible, local Gibbs distributions represent the LCL problems that can be solved by sequential local oblivious algorithms.
For any class of locally admissible, local Gibbs distributions $\Problem{M}=\{\mu_{(G,\vec{x})}\}$, the problem of constructing a $\vec{y}$ that is feasible with respect to $\mu_{(G,\vec{x})}$, can always be solved by a sequential local oblivious algorithm on any vertex ordering, i.e.~the problem is in \SLOCAL$(O(1))$.
\end{remark}


\ifabs{}{
\paragraph{The \LOCAL{} Model:}
In the \LOCAL{} model~\cite{naor1995can, peleg2000distributed}, the network is a simple, undirected graph $G=(V,E)$. Initially, each node $v\in V$ receives a local input and an arbitrary long random bit string sampled independently at $v$. For a \LOCAL{} algorithm with time complexity $t$, each node $v\in V$ gathers all information within radius $t$ from $v$, including the topology of the graph, the inputs and random bits of the nodes within that radius, and performs an arbitrary local computation with the information to compute an output.
}

\section{From Sequential to Parallel for Distributed Sampling}\label{sec:SLOCAL}
In this section, we establishes the computational equivalence (up to polylogarithmic factor) between approximate inference and approximate sampling in the \LOCAL{} model. The results in this section holds for general classes of joint distributions $\Problem{M}=\{\mu_{(G,\vec{x})}\}$.

A key step is to first resolve the problems sequentially with bounded locality.
The sequential local mode (\SLOCAL) is introduced in a recent breakthrough~\cite{ghaffari2016complexity}. 
\ifabs{}{We adopt the randomized version. 
Let $G=(V,E)$ be a simple, undirected graph with $n=|V|$. Each node $v\in V$ maintains a local state $S_v$ in its unbounded local memory. Initially, $S_v$ contains $v$'s local input and an arbitrarily long random bit string generated by $v$. 
An \SLOCAL{} algorithm $\mathcal{A}$ with locality $r(n)$ scans the nodes in an arbitrary ordering $\pi = (v_1,v_2,\ldots,v_n)$ provided by an adversary. 
When processing node $v$, $\mathcal{A}$ reads the states $S_u$ for all $u\in B_{r_v}(v)$ with $r_v\le r(n)$, then performs unbounded computation to update the state $S_v$ and compute the output $Y_v$ at $v$. 

}
An important property of the \SLOCAL{} model is that any \SLOCAL{} algorithm with  locality $r$ can be transformed into a \LOCAL{} algorithm with time complexity $O(r)$ multiplying the cost for network decomposition~\cite{panconesi1996complexity}.
We restate this in the following lemma in a slightly more refined way, in order to cover more general problems such as distributed sampling.

\begin{lemma}[Ghaffari, Kuhn, Maus~\cite{ghaffari2016complexity}]
\label{lemma-slocal-local}
Let $\mathcal{A}$ be an \SLOCAL{} algorithm which given any instance $\mathcal{I}$ on graph $G=(V,E)$ with $n=|V|$, any ordering $\pi$ of nodes in $V$, returns within locality $r(n)$ a random vector $Y=(Y_v)_{v\in V}$, where each node $v\in V$ outputs $Y_v$, such that $Y$ follows the distribution $\alg{\mu}_{\mathcal{I},\pi}$.


Then there is a \LOCAL{} algorithm $\mathcal{B}$  which given any instance $\mathcal{I}$ on graph $G=(V,E)$ with $n=|V|$, outputs within time complexity $O(r(n)\log^2 n)$ a pair $(Y_v,F_v)$ at each node $v\in V$, where $F_v$ is a Boolean random variable indicating whether the algorithm fails locally at $v$, satisfying that $\sum_{v\in V}\mathbb{E}[F_v]<\frac{1}{n^2}$ and conditioning on that $F_v=0$ for all nodes $v\in V$, the distribution of $Y=(Y_v)_{v\in V}$ is precisely $\alg{\mu}_{\mathcal{I},\pi}$ for some ordering  $\pi$ of nodes in $V$.
\end{lemma}


\ifabs{
With this lemma, the standard sequential sampler by progressively sampling according to the marginal distributions, can be transformed to a \LOCAL{} algorithm for approximate sampling.
}{
Lemma~\ref{lemma-slocal-local} is proved in the same way as Theorem 1.6 in~\cite{ghaffari2016complexity}.  Specifically, the \LOCAL{} algorithm $\mathcal{B}$ constructs an $(O(\log n), O(\log n))$-network decomposition on the power graph $G^{r+1}$, where every pair of $u$ and $v$ that $\dist_G(u,v)\le r+1$ are connected by an edge, and then simulates in parallel according to the chromatic scheduler provided by the network decomposition, the \SLOCAL{} algorithm $\mathcal{A}$ upon some ordering $\pi$ of nodes.

Lemma~\ref{lemma-slocal-local} can be proved by going through exactly the same proof, with a slightly more refined purpose to verify that all failures caused by network decomposition are locally certifiable and the output distribution $\alg{\mu}_{\mathcal{I},\pi}$ is preserved conditioning on success.





\subsection{Approximate Inference $\implies$ Approximate Sampling}
}


\begin{theorem}\label{thm:approx-infer-approx-sample}
For any class of joint distributions $\Problem{M}=\{\mu_{(G,\vec{x})}\}$, 
if there is a \LOCAL{} algorithm for approximate inference (within arbitrary total variation error $\delta>0$) with time complexity at most $t(n,\delta)$, 
then there is a \LOCAL{} algorithm for approximate sampling (within arbitrary total variation error $\delta>0$) with time complexity $O\left(t\left(n,\frac{\delta}{n}\right)\log^2 n\right)$.
\end{theorem}

\ifabs{
On the other hand, the reduction from approximate inference to approximate sampling in the \LOCAL{} model can be done by averaging over all local random choices.
}{
First, observe that for the inference problems, any randomized \LOCAL{} algorithm with certifiable local failures can be transformed into a deterministic \LOCAL{} algorithm with no failure by taking average over all random bits accessed by the randomized \LOCAL{} algorithm at each node $v$ that produce the successful output at $v$. This actually holds more generally for all problems with \concept{no symmetry}, where the correctness of the output of a node $v$ depends only on the instance but not on other nodes' outputs.

\begin{proposition}\label{prop:deterministic-approx-infer}
For any class of joint distributions $\Problem{M}=\{\mu_{(G,\vec{x})}\}$, 
if there is a randomized \LOCAL{} algorithm for approximate inference with time complexity at most $t(n,\delta)$ and failure probability $\sum_{v\in V}\mathbb{E}[F_v]<1$, 
then there is a deterministic \LOCAL{} algorithm for approximate inference with time complexity at most $t(n,\delta)$. 
\end{proposition}

With this proposition, from now on we assume without loss of generality any \LOCAL{} algorithm for approximate inference is deterministic and has no failure.

\begin{proof}[Proof of Theorem~\ref{thm:approx-infer-approx-sample}]
With Lemma~\ref{lemma-slocal-local}, it is sufficient to give an \SLOCAL{} algorithm for approximate sampling, which is quite standard with the access to marginal probabilities.

Let $(G, \vec{x}, \tau)$ be an instance, where $G=(V,E)$ and $\tau\in\Sigma^\Lambda$ is a feasible configuration on a subset $\Lambda\subseteq V$. A joint distribution $\mu=\mu_{(G,\vec{x})}$ is specified by the labeled graph $(G,\vec{x})$ and the target distribution is $\mu^\tau=\mu_{(G,\vec{x})}$.

Let $v_1,v_2,\ldots,v_n$ be an arbitrary ordering of vertices in $V$. 
The algorithm will sample a random $\sigma\in\Sigma^V$ such that $\sigma_\Lambda=\tau$, by randomly generating $\sigma(v_i)$ vertex by vertex.
For $i=1,2,\ldots,n$, let $\sigma_{i-1}$ denote the configuration over $\{v_1,v_2,\ldots,v_{i-1}\}$ that has been randomly generated so far, if $v_i\in\Lambda$ then $\sigma(v_i)=\tau(v_i)$; and if otherwise, $\sigma(v_i)$ is randomly generated according to the marginal distribution $\alg{\mu}_{v_i}^{\tau\wedge\sigma_{i-1}}$ at $v_i$ conditioning on $\sigma_{i-1}$ and $\tau$, where the marginal distribution $\alg{\mu}_{v_i}^{\tau\wedge\sigma_{i-1}}$ is computed by simulating the \LOCAL{} algorithm for approximate inference with total variation error $\frac{\delta}{n}$ on instance $(G,\vec{x},\tau\wedge\sigma_{i-1})$ at node $v_i$ with locality $t=t(n,\frac{\delta}{n})$.
It is guaranteed that
\[
\DTV{\alg{\mu}_{v_i}^{\tau\wedge\sigma_{i-1}}}{{\mu}_{v_i}^{\tau\wedge\sigma_{i-1}}}\le\frac{\delta}{n}.
\]
Finally, $\sigma=\sigma_n\in\Sigma^V$ is returned. Let $\alg{\mu}$ denote its distribution. By a coupling argument, it can be verified that $\DTV{\alg{\mu}}{\mu^{\tau}}\le\delta$.
\end{proof}


\subsection{Approximate Sampling $\implies$ Approximate Inference}
}

\begin{theorem}\label{thm:approx-sample-approx-infer}
For any class of joint distributions $\Problem{M}=\{\mu_{(G,\vec{x})}\}$, 
if there is a \LOCAL{} algorithm for approximate sampling (within arbitrary total variation error $\delta>0$) with time complexity at most $t(n,\delta)$, 
then there is a \LOCAL{} algorithm for approximate inference (within total variation error $\delta + \epsilon_0$) with time complexity at most $t(n,\delta)$, where $\epsilon_0 \ge \sum_{v\in V}\mathbb{E}[F_v]$ is the probability that the approximate sampling algorithm fails.
\end{theorem}
\ifabs{}{
\begin{proof}
Let $\mathcal{A}_{\delta}$ be the \LOCAL{} algorithm for approximate sampling with total variation error $\delta$ and time complexity at most $t(n,\delta)$. 
Given any instance $(G, \vec{x}, \tau)$, upon termination the algorithm $\mathcal{A}_\delta$ outputs a random $(y_v, F_v)$ at each node $v\in V$, where $F_v$ is a Boolean random variable indicating the local failure at node $v$. It is guaranteed that $\sum_{v \in V}\mathbb{E}[F_v] \leq \epsilon_0$ and conditioning on $\sum_{v\in V}F_v=0$, the random vector $\vec{y}=(y_v)_{v\in V}$ follows a distribution $\alg{\mu}$ such that $\DTV{\alg{\mu}}{\mu^\tau}\le \delta$ where $\mu^\tau$ is the target distribution. Therefore, for the distribution $\tilde{\mu}$ of the random vector $\vec{y}$ (without any condition), it holds that
\[
\DTV{\tilde{\mu}_v}{\mu_v^\tau}\le \DTV{\tilde{\mu}}{\mu^\tau}\le\DTV{\alg{\mu}}{\mu^\tau}+\Pr\left[\mbox{$\sum_{v\in V}F_v>0$}\right]\le \delta+\epsilon_0.
\]
The marginal distribution $\tilde{\mu}_v$ can be reconstructed at node $v$ with locality at most $t=t(n,\delta)$ by enumerating the random bits used in $\mathcal{A}_\delta$ to generate the random variable $y_v$. This gives a \LOCAL{} algorithm for approximate inference within total variation error $\delta+\epsilon_0$ and time complexity $t(n,\delta)$.
\end{proof}
}

\section{Local Self-Reductions}\label{sec:JVV}
It is well-known  that on Turing machines, for self-reducible problems the accuracy of approximate counting can be boosted and approximate counting implies exact sampling (the Jerrum-Valiant-Vazirani sampler)~\cite{JVV86}.

In this section, we give two boosting results for local distributed counting and sampling:
\begin{itemize}
\item a \concept{local boosting} for approximate inference which transforms approximate inference with bounded total variation error to the one with bounded multiplicative error;
\item a \concept{distributed JVV sampler} which uses approximate inference to achieve exact sampling via the \concept{local rejection sampling}.
\end{itemize}
Both results utilizes \concept{local  self-reductions}, the self-reductions with bounded locality. The correctness of such reductions relies on the spatial Markovian (conditional independence) property of local Gibbs distributions.
\newcommand{\ApproxCounting}{
\subsection{Approximate Inference $\iff$ Approximate Counting}\label{sec:approx-counting}
For a Gibbs distribution specified by $(G,\Sigma,\mathcal{F})$, where $G=(V,E)$, each configuration $\sigma\in\Sigma^V$ is assigned a weight $w(\sigma)$ given by~\eqref{eq:configuration-weight}. 
Recall that the partition function is given by the total weights $Z=\sum_{\sigma}w(\sigma)$. 
Furthermore, for any configuration $\tau\in\Sigma^\Lambda$ on a subset $\Lambda\subseteq V$, we define:
\[
Z(\tau)\triangleq\sum_{\sigma\in\Sigma^V: \sigma_{\Lambda}=\tau}w(\sigma).
\]

We consider the following distributed approximate counting problem. 
\begin{list}{}{}
\item[\textbf{Approximate counting:}] 
Given any instance $(G,\vec{x},\tau)$, for any $0<\epsilon<1$, 
each node $v$ returns a $\alg{z}_v$, 
such that  the product $\alg{Z}=\prod_{v\in V}\alg{z}_v$ satisfies 
$1-\epsilon\le \frac{\alg{Z}}{Z(\tau)}\le 1+\epsilon$.
\end{list}

In next theorem, we show that the approximate counting and the approximate inference are computationally equivalent up to a polylogarithmic factor in the \LOCAL{} model.

\begin{theorem}
For a class of local Gibbs distributions  $\Problem{M}=\{\mu_{(G,\vec{x})}\}$, 
if there is a \LOCAL{} algorithm for approximate inference with time complexity at most $t(n,\delta)$, 
then there is a \LOCAL{} algorithm for approximate counting with time complexity at most $O(t(n,\delta)\log^2n)$;
and conversely, 
if there is a \LOCAL{} algorithm for approximate counting with time complexity at most $t(n,\delta)$,
then there is a \LOCAL{} algorithm for approximate inference with time complexity at most $O(t(n,\delta))$.
\end{theorem}
}

\subsection{The Boosting Lemma}\label{sec:boosting}


We consider approximate inference with a stronger accuracy guarantee. 
The \concept{multiplicative error function} $\err{\cdot}{\cdot}$ is defined as follows:
for any two distributions $\mu$ and $\alg{\mu}$ over the same sample space $\Sigma$, 
\begin{align}\label{eq:error-function}
\err{{\mu}}{\alg{\mu}}\triangleq\max_{x\in\Sigma}|\ln{{\mu}(x)}-\ln{\alg{\mu}(x)}|,
\end{align}
with the convention that $0/0=1$ and $\ln 0-\ln0=\ln(0/0)=0$.
\begin{list}{}{}
\item[\textbf{Approximate inference (with multiplicative error $\epsilon$):}] 
For any instance $(G,\vec{x},\tau)$, any $0<\epsilon<1$, 
each node $v$ returns a marginal distribution $\alg{\mu}_v$, 
such that $\err{\alg{\mu}_v}{\mu_v^{\tau}}\le\epsilon$.
\end{list}
For sufficiently small $\epsilon>0$, the condition $\err{\alg{\mu}_v}{\mu_v^{\tau}}\le\epsilon$ implies
\ifabs{$1-\epsilon\approx \mathrm{e}^{-\epsilon}\le \frac{\alg{\mu}_v(c)}{\mu_v^{\tau}(c)} \le \mathrm{e}^{\epsilon} \approx 1+\epsilon$ for any $c\in\Sigma$,}{
\[
\forall c\in\Sigma:\quad 1-\epsilon\approx \mathrm{e}^{-\epsilon}\le \frac{\alg{\mu}_v(c)}{\mu_v^{\tau}(c)} \le \mathrm{e}^{\epsilon} \approx 1+\epsilon,
\]}
which gives a more accurate approximation than the bounded total variation error.

The boosting lemma stated below says that for local Gibbs distributions, approximate inference with total variation error can be boosted into that with multiplicative error.

\begin{lemma}[boosting lemma]
\label{lemma-local-boosting}
For any class of local Gibbs distributions $\Problem{M}=\{\mu_{(G,\vec{x})}\}$, 
if there is a \LOCAL{} algorithm for approximate inference (within arbitrary total variation error $\delta>0$) with time complexity at most $t(n,\delta)$, 
then there is a \LOCAL{} algorithm for approximate inference (within arbitrary multiplicative error $0<\epsilon < 1$) with time complexity $O(t(n,\frac{\epsilon}{5qn}))$, where $q=|\Sigma|\le\poly(n)$ is the size of the alphabet.
\end{lemma}
\ifabs{}{
\begin{proof}
Let $\mathcal{A}^+_\delta$ denote the \LOCAL{} algorithm for approximate inference with arbitrary total variation error $\delta>0$ whose time complexity is $t(n, \delta)$. 
We construct a \LOCAL{} algorithm $\mathcal{A}^{\times}_\epsilon$ for approximate inference with arbitrary multiplicative error $0<\epsilon< 1$. 

Let $(G, \vec{x}, \tau)$ be an instance, where $G=(V,E)$. The joint distribution $\mu=\mu_{(G, \vec{x})}$ is a Gibbs distribution given by $(G,\Sigma,\mathcal{F})$, where $q=|\Sigma|$.
Since the Gibbs distribution is local, we assume that there is an $\ell=O(1)$ such that $\forall (f, S) \in \mathcal{F}: \max_{u, v \in S}\dist_G(u, v) \leq \ell$. 

\paragraph{Algorithm $\mathcal{A}^{\times}_\epsilon$:}
Let $v\in V$. We assume $v \notin \Lambda$, otherwise the inference problem is trivial. Let $\delta=\frac{\epsilon}{5qn}$ and $t = t(n, \frac{\epsilon}{5qn})$ be the time complexity of the \LOCAL{} algorithm $A_{\delta}^+$. 
Node $v$ collects all information up to distance $2t+\ell$ and simulates the following algorithm locally.

Recall that $\Ball_r(v)$ denotes the $r$-ball centered at $v$ in $G$. 
We define
\[
\Gamma=\Ball_{t+\ell}(v)\setminus(\Ball_{t}(v)\cup\Lambda).
\]
Let $v_1,v_2,\ldots,v_m$, where $m=|\Gamma|$, be vertices in $\Gamma$ enumerated in the increasing order of their unique IDs.
A sequence of configurations $\tau_i\in\Sigma^{\Lambda_{i}}$ on subsets $\Lambda_i$, $0\le i\le m$, is constructed as follows:
\begin{itemize}
\item Initially, let $\Lambda_0 = \Lambda$ and $\tau_0 = \tau$.
\item For $i=1,2,\ldots,m$, let $\Lambda_i = \Lambda_{i - 1} \cup \{v_i\}$, and the configurations $\tau_i\in\Sigma^{\Lambda_i}$ is constructed such that $\tau_i$ is consistent with $\tau_{i-1}$ over $\Lambda_{i-1}$, and $\tau_i(v_i)=c_i$ for the $c_i\in\Sigma$ that maximizes the marginal probability $\alg{\mu}_{v_i}^{\tau_{i-1}}(c_i)$ where $\alg{\mu}_{v_i}^{\tau_{i-1}}$ is the marginal distribution returned by $\mathcal{A}^{+}_\delta$ at node $v_i$ on the instance $(G,\vec{x},\tau_{i-1})$. 
\end{itemize}
Finally, the marginal distribution $\mu^{\tau_m}_v$ is returned. 
Due to the conditional independence guaranteed by Proposition~\ref{prop:cond-ind}, $\mu^{\tau_m}_v$ is fully determined by the information in $\Ball_{t+\ell}(v)$.
Specifically, denoted $B=B_{t+\ell}(v)$, and define $\mathcal{C}$ to be the set of all configurations $\sigma\in\Sigma^B$ consistent with $\tau_m$ over $B\cap \Lambda_m$, i.e.
\[
\mathcal{C} = \{\sigma \in \Sigma^B \mid \forall u \in B \cap (\Gamma \cup \Lambda): \sigma_u =\tau_m(u) \}.
\]
Then for every $c\in \Sigma$, the marginal probability $\mu^{\tau_m}_v(c)$ is computed as 
\[
\mu^{\tau_m}_v(c)=\frac{\sum_{\sigma\in \mathcal{C} :\sigma_v=c}w_B(\sigma)}{\sum_{\sigma\in \mathcal{C}}w_B(\sigma)},
\]
where $w_B(\sigma)=\prod_{(f,S)\in\mathcal{F}: S\subseteq B}f(\sigma_S)$.
This finishes the definition of Algorithm $\mathcal{A}^{\times}_\epsilon$.


\paragraph{}
We then show that ${\tau_m}$ is feasible with respect to $\mu$, so the marginal distribution $\mu^{\tau_m}_v$ is well-defined. Furthermore, it holds that
\begin{align}
\label{eq-taum-tau}
\forall c \in \Sigma: \qquad \mathrm{e}^{-\epsilon}\mu^{\tau}_v(c)  \leq \mu^{\tau_m}_v(c)\leq \mathrm{e}^\epsilon \mu^{\tau}_v(c),
\end{align}
which proves the Theorem.

For the sequence of configurations $\tau_i\in\Sigma^{\Lambda_i}$ on subsets $\Lambda_i$, $0\le i\le m$, constructed in Algorithm $\mathcal{A}^\times_\epsilon$, for every $c \in \Sigma$, let $\tau_i^c\in\Sigma^{\Lambda_i\cup\{v\}}$ denote the configuration on $\Lambda_i\cup\{v\}$ such that $\tau_i^c$ is consistent with $\tau_i$ over $\Lambda_i$ and $\tau_i^c(v)=c$. 

\begin{claim*}
If $\tau^c_0$ is feasible with respect to $\mu$, then all $\tau^c_i$ are feasible  with respect to $\mu$. 
\end{claim*}
\begin{proof}
We prove this by induction on $i$.
For $i = 0$, the claim holds trivially.
For general $i$, suppose that $\tau^c_{i-1}$ is feasible. Then $\tau_{i-1}$ must be feasible in the first place because $\tau^c_{i-1}$ extends $\tau_{i-1}$.
The two configurations $\tau^c_{i-1}$ and $\tau_{i-1}$ differ only at vertex $v$ and $\dist_G(v, v_{i-1}) > t$, where $t=t(n,\frac{\epsilon}{5qn})$, Algorithm $\mathcal{A}^{+}_{\delta}$ will output the same marginal distribution $\alg{\mu}_{v_{i-1}}^{\tau_{i-1}}$ at node $v_{i-1}$ on the instances $(G,\vec{x},\tau^c_{i-1})$ and $(G,\vec{x},\tau_{i-1})$, such that $\DTV{\alg{\mu}_{v_{i-1}}^{\tau_{i-1}}}{\mu^{\tau_{i-1}^c}_{v_{i}}}\le\frac{\epsilon}{5qn}$ and $\DTV{\alg{\mu}_{v_{i-1}}^{\tau_{i-1}}}{\mu^{\tau_{i-1}}_{v_{i}}}\le\frac{\epsilon}{5qn}$. By triangle inequality,
\begin{align}
\label{eq-dtv-tau-tauc}
\DTV{\mu^{\tau_{i-1}^c}_{v_{i}}}{\mu^{\tau_{i-1}}_{v_{i}}}\leq \frac{2\epsilon}{5nq}.
\end{align}
Recall that $\tau_i$ is constructed from $\tau_{i-1}$ in such a way that  $\tau_i(v_i)=c_i$ for the $c_i\in\Sigma$ that maximizes the marginal probability $\alg{\mu}_{v_i}^{\tau_{i-1}}(c_i)$, and $q=|\Sigma|$.
Therefore, we have
\begin{align}
\label{eq-lower-bound-tau}
\mu^{\tau_{i-1}}_{v_{i}}(\tau_i(v_i)) \geq \frac{1}{q} - \frac{\epsilon}{5nq}.
\end{align}
Note that $\mu_{\Lambda_i\cup\{v\}}(\tau_{i}^c)= \mu_{\Lambda_{i-1}\cup\{v\}}(\tau_{i-1}^c)\cdot \mu_{v_i}^{\tau_{i-1}^c}(c_i)\ge \mu_{v_i}^{\tau_{i-1}^c}(c_i)$, where $c_i=\tau_i(v_i)$. Combining with~\eqref{eq-dtv-tau-tauc} and~\eqref{eq-lower-bound-tau}, we have
\begin{align*}
\mu_{\Lambda_i\cup\{v\}}(\tau_{i}^c)\ge
\mu_{v_i}^{\tau_{i-1}^c}(c_i) \geq 	\mu^{\tau_{i-1}}_{v_{i}}(c_i) - \frac{2\epsilon}{5nq} \geq \frac{1}{q} - \frac{3\epsilon}{5nq} > 0,
\end{align*}
which implies that $\tau_{i}^c$ is feasible with respect to $\mu$. The claim is proved.
\end{proof}

Recall that $\tau_0 = \tau$ for a feasible $\tau\in\Sigma^{\Lambda}$. There must exist $c\in \Sigma$ such that $\tau_0^c$ is feasible, which according to above claim implies that all $\tau_i^c$ are feasible, and hence in particular, $\tau_m^c$ is feasible, which means $\tau_m$ is feasible in the first place since $\tau_m^c$ extends $\tau_m$.

Consider each $c \in \Sigma$ that $\mu^\tau_v(c) > 0$. For such $c\in\Sigma$,  $\tau_{0}^c$ must be feasible since $\tau_0 = \tau$, which according to above claim implies that $\tau_{0}^c$ is feasible. Denote that $c_i=\tau_i(v_i)$ for $0\le i\le m$. Apply the chain rule in two different orders of vertices, we have
\begin{align*}
\mu^{\tau}_{\Lambda\cup\{v\}}(\tau_m^c) 
&= \left(\prod_{i=1}^m \mu_{v_i}^{\tau_{i-1}}(c_i)\right)\mu^{\tau_m}_v(c), \\
\text{and}\quad\mu^{\tau}_{\Lambda\cup\{v\}}(\tau_m^c)   
 &= \mu^\tau_v(c) \left(\prod_{i = 1}^m \mu_{v_i}^{\tau^c_{i-1}}(c_i)\right).
 \end{align*}
Solving the above equations gives us
\begin{align*}
\mu^{\tau_m}_v(c) &= \left( \prod_{i=1}^m\frac{\mu_{v_i}^{\tau^c_{i-1}}(c_i)}{\mu_{v_i}^{\tau_{i-1}}(c_i)} \right)\mu^\tau_v(c).
\end{align*}
For each $1\leq i \leq m$, by~\eqref{eq-dtv-tau-tauc} and~\eqref{eq-lower-bound-tau}, it holds that
\begin{align*}
\mathrm{e}^{-\epsilon / n}
\le
1 - \frac{2\epsilon}{5n-\epsilon} 
\leq \frac{\mu_{v_i}^{\tau^c_{i-1}}(c_i)}{\mu_{v_i}^{\tau_{i-1}}(c_i)} 
\leq 
1 + \frac{2\epsilon}{5n-\epsilon}
\le \mathrm{e}^{\epsilon / n}.
\end{align*}
Since $m=|\Lambda| \leq n$,  we have
\begin{align*}
\mathrm{e}^{-\epsilon}  \leq \frac{\mu^{\tau_m}_v(c)}{\mu^{\tau}_v(c)}\leq \mathrm{e}^\epsilon .
\end{align*}
For those $c \in \Sigma$ that $\mu_v^\tau(c) = 0$, it holds that $\tau_0^c$ is infeasible. 
It must hold that $\mu_v^{\tau_m}(c) = 0$, since if otherwise $\mu_v^{\tau_m}(c) > 0$, then $\tau_m^c$ is feasible,  contradicting that $\tau_0^c$ is infeasible since $\tau_m^c$ extends $\tau_0^c$.
\end{proof}
}

\subsection{Distributed JVV Sampler}

The Jerrum-Valiant-Vazirani (JVV) sampler~\cite{JVV86} is a general global reduction  from exact sampling to approximate counting for self-reducible problems via rejection sampling. Here we give a local distributed JVV sampler by realizing a local rejection sampling.
\begin{theorem}\label{thm:approx-infer-exact-sample}
For a class of local Gibbs distributions $\Problem{M}=\{\mu_{(G,\vec{x})}\}$, 
if there is a \LOCAL{} algorithm for approximate inference (with total variation error  $\le {1}/{5qn^4}$) with time complexity at most $t(n)$, where $q=|\Sigma|\le\poly(n)$ is the size of the alphabet, 
then there is a \LOCAL{} algorithm for exact sampling with time complexity  $O(t(n)\log^2n)$.
\end{theorem}
By the boosting lemma (Lemma~\ref{lemma-local-boosting}), 
the algorithm for approximate inference with total variation error in the assumption of Theorem~\ref{thm:approx-infer-exact-sample} can be boosted to an  approximate inference algorithm with multiplicative error ${1}/{n^3}$ and time complexity $O(t(n))$.
Theorem~\ref{thm:approx-infer-exact-sample} is then a consequence of the following proposition.

\begin{proposition}
\label{lemma-multi-infer-to-sample}
For a class of local Gibbs distributions $\Problem{M}=\{\mu_{(G,\vec{x})}\}$, 
if there is a \LOCAL{} algorithm for approximate inference (with multiplicative error  $\le {1}/{n^3}$) with time complexity at most $t(n)$, 
then there is a \LOCAL{} algorithm for exact sampling with time complexity $O(t(n)\log^2n)$.
\end{proposition}

Let $\mathcal{A}$ be the \LOCAL{} algorithm for approximate inference with multiplicative error $1 /n^3 $ and time complexity at most $t=t(n)$. 
We construct an \SLOCAL{} algorithm called \JVV{} for exact sampling with time complexity  $O(t)$. \ifabs{}{For convenience, the \JVV{} algorithm is presented as a multi-pass \SLOCAL{} algorithm with the ability of writing nearby nodes' internal memories. It was observed in~\cite{ghaffari2016complexity} that these variations will not substantially change the power of \SLOCAL{} algorithms.


\begin{lemma}\label{lemma-SLOCAL-variation}
The followings hold for \SLOCAL{} algorithms.
\begin{enumerate}
\item (Observation 2.1 in~\cite{ghaffari2016complexity})
Any \SLOCAL{} algorithm $\mathcal{A}$ with locality $R$ in which each node $v$ can write into the local memory $S_u$ of other nodes $u$ within its radius $r \leq R$ can be transformed into an \SLOCAL{} algorithm $\mathcal{B}$ with locality $r + R$ in which $v$ writes only in its own memory $S_v$.
\item (Lemma 2.2 in~\cite{ghaffari2016complexity})
Any $k$-pass \SLOCAL{} algorithm $\mathcal{A}$ with locality $r_i$ in the $i$-th pass for $i = 1 ,\ldots, k$, can be transformed into a single-pass \SLOCAL{} algorithm $\mathcal{B}$ with locality $r_1+2\sum_{i = 2}^k r_i$.
\end{enumerate}
\end{lemma}

}

Let $(G, \vec{x}, \tau)$ be an instance, where $G=(V,E)$. The joint distribution $\mu=\mu_{(G, \vec{x})}$ is a Gibbs distribution given by $(G,\Sigma,\mathcal{F})$, where $q=|\Sigma|$.
Since the Gibbs distribution is local, we assume that there is an $\ell=O(1)$ such that $\forall (f, S) \in \mathcal{F}: \max_{u, v \in S}\dist_G(u, v) \leq \ell$. 
The $\tau\in\Sigma^\Lambda$ is an arbitrary feasible configuration on an arbitrary subset $\Lambda\subseteq V$. The distribution $\mu^\tau$ is the target distribution that we want to sample from.

\ifabs{\paragraph{The \JVV{} algorithm.}}{
\subsubsection{The \JVV{} algorithm}}
The \JVV{} algorithm is an \SLOCAL{} algorithm for the \concept{local rejection sampling}: Upon termination, the algorithm returns a $(Y,F')=((Y_{v})_{v\in V}, (F'_{v})_{v\in V})$, where 
each node $v\in V$ outputs $Y_{v}$ and $F'_{v}$, such that $Y\in\Sigma^V$ is a random configuration, and each $F'_{v}\in\{0,1\}$  indicates that the algorithm fails (rejects) locally at node $v$.

The \SLOCAL{} algorithm consists of three passes. 
In each pass, the algorithm scans the nodes in the same ordering $\pi = v_1,v_2,\ldots,v_n$ provided to the algorithm by an adversary.

In the first pass, a configuration $\sigma_0\in\Sigma^V$ that is feasible with respect to the target distribution $\mu^\tau$, called the \concept{ground state}, is constructed by the following procedure:
\begin{itemize}
\item Initially, $\sigma_0\in\Sigma^{\Lambda}$ and $\sigma_0 = \tau$. Then $\sigma_0$ is updated at each step.
\item In the $i$-th step, simulate the algorithm $\mathcal{A}$ at node $v_i$ within radius $t=t(n)$ on the instance $(G, \vec{x}, \sigma_0)$, which returns a marginal distribution $\alg{\mu}^{\sigma}_{v_i}$. Pick an arbitrary $c_i\in \Sigma$ with $\alg{\mu}^{\sigma_0}_{v_i}(c_i)>0$, and extend the current $\sigma_0$ further onto $v_i$ by setting $\sigma_0(v_i) \gets c_i$.
\end{itemize}

In the second pass, a random configurations $Y\in\Sigma^V$ is generated independently by the following procedure:
\begin{itemize}
\item Initially, $Y\in\Sigma^\Lambda$ and $Y = \tau$. Then $Y$ is randomly updated at each step.
\item In the $i$-th step, simulate the algorithm $\mathcal{A}$ at node $v_i$ on the instance $(G, \vec{x}, Y)$, which returns a marginal distribution $\alg{\mu}^{Y}_{v_i}$. Sample a random $Y_i\in \Sigma$ independently according to $\alg{\mu}^{Y}_{v_i}$, and extend the current $Y$ further onto $v_i$ by setting $Y(v_i) \gets Y_i$.
\end{itemize}
Let $\alg{\mu}^\tau$ be the distribution of $Y\in\Sigma^V$ generated as above. It closely approximates the $\mu^\tau$.
\ifabs{}{
\begin{claim}\label{claim-JVV-approx}
For any $\sigma \in \Sigma^V$,	it holds that $\mathrm{e}^{-1/n^2} \leq \frac{\alg{\mu}^\tau(\sigma)}{\mu^\tau(\sigma)} \leq \mathrm{e}^{1/n^2}$.
\end{claim}
}

In the third pass, each node $v_i$ carefully computes a probability $q_{v_i}$,
and uses $q_{v_i}$ to independently sample a random $F'_{v_i}\in\{0,1\}$ indicating the local failure at $v_i$.
\begin{itemize}
\item Initially, let $\sigma_0\in\Sigma^{V}$ be the ground state constructed in the first pass.
\item 
A sequence of configurations $\sigma_1,\sigma_2,\ldots,\sigma_n\in\Sigma^V$ is constructed, where $\sigma_n=Y$ is the configuration randomly generated in the second pass, such that the following invariants hold for all $0\le i\le n$:
\begin{align}
&\sigma_i\text{ is feasible with respect to }\mu^\tau;\label{eq:JVV-path-invariant-0}\\
&\forall 1\le j\le i: 
\quad\sigma_i(v_j)=Y(v_j);\label{eq:JVV-path-invariant-1}\\ 
&\sigma_i\text{ and }\sigma_{i-1}\text{ differ only over the $t$-ball }B_t(v_i).\label{eq:JVV-path-invariant-2}
\end{align}
The invariants hold trivially for $\sigma_0$. 
In the $i$-th step, assume the above invariants for $\sigma_{i-1}$.
Construct a $\sigma_i\in\Sigma^V$ satisfying these invariants by enumerating all configurations $\sigma'\in\Sigma^{B}$ over the $t$-ball $B=B_t(v_i)$ and trying replacing the assignment of $\sigma_{i-1}(B)$ with $\sigma'$.
The following claim guarantees the existence of such $\sigma_i\in\Sigma^V$.
\begin{claim}\label{claim-JVV-path}
Assume that $\sigma_{i-1}\in\Sigma^V$ is feasible with respect to $\mu^\tau$ and $\sigma_{i-1}(v_j)=Y(v_j)$ for all $1\le j\le i-1$. There exists a $\sigma_{i}\in\Sigma^V$ satisfying~\eqref{eq:JVV-path-invariant-0}, \eqref{eq:JVV-path-invariant-1} and~\eqref{eq:JVV-path-invariant-2}.
\end{claim}
Assuming this claim, the algorithm can always find a good $\sigma_i$ and verify the invariants within bounded radius from $v_i$. Specifically, assuming the invariants for $\sigma_{i-1}$, the invariants~\eqref{eq:JVV-path-invariant-1} and~\eqref{eq:JVV-path-invariant-2} can be verified at $v_i$ within radius $t$, and~\eqref{eq:JVV-path-invariant-0} can be verified within radius $t+\ell$\ifabs{}{ due to the conditional independence property stated in Proposition~\ref{prop:cond-ind}}.
Once a good $\sigma_i$ is found, node $v_i$ updates the internal states of all nodes in $B_t(v_i)$ to update the current configuration to $\sigma_i$.

Node $v_i$ computes the value of 
\begin{align}
\label{eq-definition-qv}
q_{v_i} =q_{v_i}(Y)\triangleq \frac{\alg{\mu}^{\tau}(\sigma_{i-1})w(\sigma_i)}{\alg{\mu}^{\tau}(\sigma_i)w(\sigma_{i-1})}	\mathrm{e}^{-3/n^2},
\end{align}
where $\alg{\mu}^{\tau}$ stands for the distribution of the random configuration $Y$ generated in the second pass, and $w(\cdot)$ is the weight for the Gibbs distribution defined in\ifabs{ Definition~\ref{def:Gibbs}}{~\eqref{eq:configuration-weight}}.

\begin{claim}
\label{claim-qv}	
The $q_{v_i}$ defined in~\eqref{eq-definition-qv} can be computed at $v_i$ within radius $3t+\ell=O(t)$  and it always holds that $\mathrm{e}^{-5/n^2} \leq q_{v_i} \leq 1$.
\end{claim}
At last, $v_i$ samples a random $F'_{v_i}\in\{0,1\}$ such that $F'_{v_i}=1$ with probability $q_{v_i}$.
\end{itemize}
Finally, each node $v_i$ returns a pair $(Y_{v_i},F'_{v_i})$ where $Y=(Y_{v_i})_{1\le i\le n}$ is the random configuration sampled in the second pass and the algorithm fails locally at $v_i$ if $F'_{v_i}=1$.

\ifabs{}{
\subsubsection{Proofs of the three claims}
\begin{proof}[Proof of Claim~\ref{claim-JVV-approx}]
Let $Y^i$ denote the partially specified random configuration $Y$ constructed in the $i$-th step in the second pass of the the \JVV{} algorithm.
For any $\sigma \in \Sigma^V$, let $\sigma^{i}$ denote the $Y^i$ when $Y=\sigma$, where $\sigma^0=\tau$ and $\sigma^n=Y=\sigma$.
Clearly, for the distribution $\alg{\mu}^\tau$ of $Y$, 
\begin{align}
\alg{\mu}^\tau(\sigma) = \prod_{i = 1}^n \alg{\mu}^{\sigma^{i-1}}_{v_i}(\sigma_{v_i}).\label{eq-chain-rule-alg}
\end{align}
On the other hand, apply the chain rule to the Gibbs distribution $\mu^\tau$. We have
\begin{align*}
\mu^\tau(\sigma) = \prod_{i = 1}^n \mu^{\sigma^{i-1}}_{v_i}(\sigma_{v_i}).	
\end{align*}
Since each marginal probability $\alg{\mu}^{\sigma^{i-1}}_{v_i}(\sigma_{v_i})$ is computed by an approximate inference algorithm $\mathcal{A}$ with multiplicative error $1/n^3$, it holds that
\begin{align*}
\mathrm{e}^{-1/n^3} \leq	\frac{\alg{\mu}^{\sigma^{i-1}}_{v_i}(\sigma_{v_i})}{\mu^{\sigma^{i-1}}_{v_i}(\sigma_{v_i})} \leq \mathrm{e}^{1/n^3}.
\end{align*}
The claim follows.
\end{proof}

\begin{proof}[Proof of Claim~\ref{claim-JVV-path}]
Assume that $\sigma_{i-1}\in\Sigma^V$ is feasible with respect to $\mu^\tau$ and $\sigma_{i-1}(v_j)=Y(v_j)$ for all $1\le j\le i-1$. 
We show that there exists a $\sigma_{i}\in\Sigma^V$ satisfying~\eqref{eq:JVV-path-invariant-0}, \eqref{eq:JVV-path-invariant-1}, and \eqref{eq:JVV-path-invariant-2}. 
Define set of nodes
\begin{align*}
\Gamma = \{v_j \in V \mid v_j \notin B_{t}(v) \text{ or } j \leq i - 1\}.	
\end{align*}
Let $\tau_1 = \sigma_{i-1}(\Gamma)$ and $\tau_2 = Y(\Gamma)$ be two configurations on subset $\Gamma$. 
By assumption and Claim~\ref{claim-JVV-approx}, $\sigma_{i-1}$ and $Y$ are feasible, so $\tau_1$ and $\tau_2$ are also feasible.

The marginal distributions returned by $\mathcal{A}$ at node $v_i$ on the two instances $(G, \vec{x}, \tau_1)$ and $(G,\vec{x}, \tau_2)$ are identical, because $\tau_1$ and $\tau_2$ agree with each other over $B_t(v_i)$. We denote this marginal distribution as $\alg{\mu}^{\tau_1}_{v_i}=\alg{\mu}^{\tau_2}_{v_i}$, which is an approximation of the marginal distributions $\mu^{\tau_1}_{v_i}$, $\mu^{\tau_2}_{v_i}$ with multiplicative error ${1/n^3}$. It holds that
\begin{align*}
\mathrm{e}^{1/n^3}\cdot\mu^{\tau_1}_{v_i}(Y(v_i)) \ge 
\alg{\mu}^{\tau_1}_{v_i}(Y(v_i)) 
&\geq \mathrm{e}^{-1/n^3}\cdot\mu^{\tau_2}_{v_i}(Y(v_i)) > 0.
\end{align*}
where the last inequality is due to that $Y$ is feasible.

This shows that $\mu^{\tau_1}_{v_i}(Y(v_i))>0$, which means there exists a feasible configuration $\sigma\in\Sigma^V$ such that $\sigma(\Gamma)=\tau_1 = \sigma_{i-1}(\Gamma)$ and $\sigma(v_i) = Y({v_i})$. Such $\sigma$ is the $\sigma_i$ we want. It can be easily verified that it satisfies the invariants~\eqref{eq:JVV-path-invariant-0}, \eqref{eq:JVV-path-invariant-1} and~\eqref{eq:JVV-path-invariant-2}.
\end{proof}

\begin{proof}[Proof of Claim~\ref{claim-qv}]
Recall that $q_{v_i}$ is defined as
\begin{align*}
q_{v_i} = \frac{\alg{\mu}^{\tau}(\sigma_{i-1})w(\sigma_i)}{\alg{\mu}^{\tau}(\sigma_i)w(\sigma_{i-1})}	\mathrm{e}^{-3/n^2}.
\end{align*}
Due to Claim~\ref{claim-JVV-path}, every $\sigma_i$ is feasible with respect to $\mu^\tau$, so $w(\sigma_i) > 0$ and moreover by Claim~\ref{claim-JVV-approx}, we have $\alg{\mu}^\tau(\sigma_i) > 0$. The ratio $q_{v_i}$ is well-defined.

Consider the configurations $\sigma_0, \sigma_1, \ldots,\sigma_n\in\Sigma^V$ constructed in the third pass of the \JVV{} algorithm. For each $\sigma_i$, we define a sequence of configurations $\sigma_i^0,\sigma_i^1,\ldots,\sigma_i^n$, where $\sigma_i^0=\tau$ and $\sigma_i^n=\sigma_i$. 
Let $\sigma_i^j$ denote the $Y^j$ when $Y=\sigma_i$, where $Y^j$ denotes the partially specified random configuration $Y$ constructed in the $j$-th step in the second pass of the the \JVV{} algorithm. Due to~\eqref{eq-chain-rule-alg}, we have
\begin{align*}
 \frac{\alg{\mu}^\tau(\sigma_{i-1})}{\alg{\mu}^\tau(\sigma_i)} = \prod_{j=1}^n\frac{\alg{\mu}^{\sigma_{i-1}^{j-1}}_{v_j}(\sigma_{i-1}(v_j))}{\alg{\mu}^{\sigma_{i}^{j-1}}_{v_j}(\sigma_{i}(v_j))}.
\end{align*}
Note that the two configurations $\sigma_i$ and $\sigma_{i-1}$ differ only on the $t$-ball $B_t(v_i)$. Therefore, for all $\sigma_i$ and $\sigma_{i-1}$ agree with each other over $B_{t}(v_k)$ for any $v_k \notin B_{2t}(v_i)$.
The two instances $(G, \vec{x}, \sigma_{i-1}^{k-1})$ and $(G, \vec{x}, \sigma_{i}^{k-1})$  are indistinguishable to the $t$-local approximate inference algorithm $\mathcal{A}$, thus
\[
\alg{\mu}^{\sigma_{i-1}^{k-1}}_{v_k}(\sigma_{i-1}(v_k)) = \alg{\mu}^{\sigma_{i}^{k-1}}_{v_k}(\sigma_{i}(v_k)),
\]
which means
\begin{align}
\label{eq-local-mu}
\frac{\alg{\mu}^\tau(\sigma_{i-1})}{\alg{\mu}^\tau(\sigma_i)} = 	\prod_{v_j \in B_{2t}(v_i)}\frac{\alg{\mu}^{\sigma_{i-1}^{j-1}}_{v_j}(\sigma_{i-1}(v_j))}{\alg{\mu}^{\sigma_{i}^{j-1}}_{v_j}(\sigma_{i}(v_j))},
\end{align}
where each marginal probability $\alg{\mu}_{v_j}^{\cdot}(\cdot)$ can be computed by the local algorithm $\mathcal{A}$ at node $v_j$ within radius $t$.

On the other hand, since the two configurations $\sigma_i$ and $\sigma_{i-1}$ differ only on nodes in $B_t(v_i)$, by the definition of the weight $w(\cdot)$ in~\eqref{eq:configuration-weight}, we have
\begin{align}
\label{eq-local-w}
\frac{w(\sigma_i)}{w(\sigma_{i-1})} = \prod_{(S, f) \in \mathcal{F} \atop S \subseteq B_{t+\ell}(v_i)}\frac{f(\sigma_i(S))}{f(\sigma_{i-1}(S))}.
\end{align}

Equations~\eqref{eq-local-mu} and~\eqref{eq-local-w} imply that the quantity $q_{v_i}$ can be computed at node $v_i$ by gathering all information up to distance $3t+\ell$.

For the Gibbs distribution $\mu$, it always holds that
$\frac{\mu^\tau(\sigma_{i-1})w(\sigma_i)}{\mu^\tau(\sigma_i)w(\sigma_{i-1})}= 1$ since the measure $\mu$ is defined proportional to the weights $w(\cdot)$.
Due to Claim~\ref{claim-JVV-approx}, we have
\begin{align*}
	\mathrm{e}^{-2/n^2} \leq \frac{\alg{\mu}^\tau(\sigma_{i-1})w(\sigma_i)}{\alg{\mu}^\tau(\sigma_i)w(\sigma_{i-1})} \leq \mathrm{e}^{2/n^2},
\end{align*}
which implies $\mathrm{e}^{-5/n^2} \leq q_{v_i}(\sigma) \leq \mathrm{e}^{-1/n^2} \leq 1$. 
\end{proof}

\subsubsection{Proof of Proposition~\ref{lemma-multi-infer-to-sample}}}
The correctness of the \JVV{} algorithm is guaranteed by the following lemma.
\begin{lemma}\label{lemma-local-rejection-sampling}
On any ordering of nodes,
the \JVV{} algorithm fails with probability at most $\sum_{i=1}^n\mathbb{E}[F'_{v_i}]=1-O(1/n)$ and conditioning on that $F'_{v_i}=0$ for all $v_i$, the $Y\in\Sigma^V$ randomly sampled by the algorithm follows precisely the target distribution $\mu^\tau$.
\end{lemma}

\ifabs{Proposition~\ref{lemma-multi-infer-to-sample} is then proved by applying Lemma~\ref{lemma-slocal-local}.}{
\begin{proof}

Let $\mathcal{G}$ denote the event that the \JVV{} algorithm succeeds, i.e.~$F'_{v_i}=0$ for all $v_i$. Recall that for any $\sigma\in\Sigma^V$,
\[
\Pr[F'_{v_i}=1\mid Y=\sigma]=q_{v_i}(\sigma),
\] 
where $q_{v_i}(\sigma)=q_{v_i}(Y)\big|_{Y=\sigma}$ is as defined in~\eqref{eq-definition-qv}. Due to Claim~\ref{claim-qv}, 
for any $\sigma\in\Sigma^V$,
\begin{align*}
 \Pr[\mathcal{G}\mid Y=\sigma] = \prod_{i=1}^nq_{v_i}(\sigma) \geq \mathrm{e}^{-5/n} = 1 - O\left(\frac{1}{n}\right).
\end{align*}
This proves that the algorithm succeeds with probability $\Pr[\mathcal{G}]=1 - O\left(\frac{1}{n}\right)$.

Next, we show that $\Pr[Y=\sigma\mid \mathcal{G}]=\mu^\tau(\sigma)$ for any $\sigma \in\Sigma^V$, which proves the lemma.

For any $\sigma \in \Sigma^V$ that $\mu^\tau(\sigma) = 0$, by Claim~\ref{claim-JVV-approx}, we have $\alg{\mu}^\tau(\sigma)=\Pr[Y=\sigma]=0$, therefore $\Pr[Y=\sigma\mid \mathcal{G}]=0$.

For any $\sigma \in \Sigma^V$ that $\mu^\tau(\sigma) > 0$, by Claim~\ref{claim-JVV-approx}, we have $\alg{\mu}^\tau(\sigma)=\Pr[Y=\sigma]>0$.
The probability  that the algorithm succeeds and outputs $Y = \sigma$ is given by
\begin{align*}
\Pr[Y =  \sigma \land \mathcal{G}] &= \alg{\mu}^\tau(\sigma)\Pr[\mathcal{G} \mid Y=\sigma]\\
&=\alg{\mu}^\tau(\sigma)\prod_{i=1}^nq_{v_i}(\sigma)\\
&=
\alg{\mu}^\tau(\sigma)\prod_{i=1}^n
\left.\left( \frac{\alg{\mu}^\tau(\sigma_{i-1})w(\sigma_i)}{\alg{\mu}^\tau(\sigma_i)w(\sigma_{i-1})}\mathrm{e}^{-3/n^2} \right)\right|_{\sigma_n=Y=\sigma}\\
&=
\left(  \frac{\alg{\mu}^\tau(\sigma_0)}{w(\sigma_0)}\cdot \mathrm{e}^{-3/n}\right)w(\sigma).
\end{align*}
Note that the factor $\left(  \frac{\alg{\mu}^\tau(\sigma_0)}{w(\sigma_0)}\cdot \mathrm{e}^{-3/n}\right)$ is independent of $\sigma$. Therefore, 
\begin{align*}
\Pr[Y =  \sigma \mid \mathcal{G}] 
&=\frac{\Pr[Y =  \sigma \land \mathcal{G}] }{\sum_{\sigma':\,\mu^{\tau}(\sigma')>0}\Pr[Y =  \sigma' \land \mathcal{G}]  }\\
&=\frac{w(\sigma)}{\sum_{\sigma':\,\mu^{\tau}(\sigma')>0}w(\sigma')}\\
&=\mu^{\tau}(\sigma).
\end{align*}
\end{proof}

Finally, we are going to prove Proposition~\ref{lemma-multi-infer-to-sample}.
It can be easily verified that the \JVV{} algorithm has locality $O(t)$ where $t=t(n)$ is the time complexity of the \LOCAL{} algorithm $\mathcal{A}$.

We apply Lemma~\ref{lemma-SLOCAL-variation} and Lemma~\ref{lemma-slocal-local} to transform the \JVV{} algorithm defined in the \SLOCAL{} model to a \LOCAL{} algorithm. Due to Lemma~\ref{lemma-slocal-local}, upon successful termination, the \LOCAL{} algorithm preserves the distribution of the output $(Y,F')$ of the \SLOCAL{} algorithm on some ordering of nodes, where $Y=(Y_v)_{v\in V}$ is the randomly sampled configuration and $F'=(F'_v)_{v\in V}$ is the vector of random indicators for local failures generated in the final pass of the \JVV{} algorithm. 
The transformation to the \LOCAL{} model will introduce another random failure $F_v''\in\{0,1\}$ to each node $v$ which is independent of $(Y,F')$, where $F''_v=1$ indicates the failure at node $v$ caused by the network decomposition and $\sum_{v\in V}\mathbb{E}[F_v'']=O(1/n^2)$. 
We combine the two failures and define $F_v=F_v'\vee F_v''$ to indicate the failure of the \LOCAL{} algorithm at node $v$. Clearly, 
\[
\sum_{v\in V}\mathbb{E}[F_v]\le \sum_{v\in V}\mathbb{E}[F_v']+\sum_{v\in V}\mathbb{E}[F_v'']=O\left(\frac{1}{n}\right).
\] 
Furthermore, since $(F_v'')_{v\in V}$ is independent of $(Y,F')$, it still holds that conditioning on that $F_v=0$ for all nodes $v\in V$, the distribution of $Y$ is precisely the target distribution.
}

\section{Approximate Inference and Strong Spatial Mixing}\label{sec:SSM}
In this section, we explore the intrinsic relation between distributed sampling/counting problems and  decays of correlation in joint distributions. The main result in this section holds for locally admissible, local Gibbs distributions. 

An important decay of correlation property for joint distributions is the \concept{strong spatial mixing}. We adopt the definition of strong spatial mixing in~\cite{weitz2006counting} into our context.
\begin{definition}[strong spatial mixing]\label{def:ssm}
Let $\delta_n:\mathbb{N}\to \mathbb{R}_{\ge0}$ be a sequence of non-increasing functions. A class of joint distributions $\Problem{M}=\{\mu_{(G,\vec{x})}\}$ is said to exhibit \emph{strong spatial mixing} with rate $\delta_n(\cdot)$ if for every distribution $\mu_{(G,\vec{x})}\in\Problem{M}$ over $\Sigma^V$, where $G=(V,E)$ and $n=|V|$,
for every $v\in V$, $\Lambda\subseteq V$, and any two feasible configurations $\sigma,\tau\in\Sigma^{\Lambda}$,
\begin{align}
\DTV{\mu_v^{\sigma}}{\mu_v^{\tau}}\le\delta_{n}(\dist_G(v,D)),\label{eq:ssm}
\end{align}
where $D\subseteq \Lambda$ is the subset on which $\sigma$ and $\tau$ differ.

In particular, the strong spatial mixing is said to be \concept{with exponential decay at rate $\alpha$}, for some $0<\alpha<1$, if the mixing rate $\delta_{\cdot}(\cdot)$ is in the form $\delta_n(t)=\poly(n)\cdot \alpha^t$.
\end{definition}

The strong spatial mixing is intrinsically related to the approximate inference in the \LOCAL{} model. In fact, the strong spatial mixing can be thought as a weaker form of  approximate inference in the \LOCAL{} model, where every node knows the graph $G$ and distribution $\mu$ but not the partially specified feasible configuration $\tau$. Therefore, it is quite natural that the approximate inference always implies strong spatial mixing. Meanwhile, the converse also holds for locally admissible, local Gibbs distributions. 

\begin{theorem}\label{thm:ssm-approx-infer}
For any class of joint distributions $\Problem{M}=\{\mu_{(G,\vec{x})}\}$, 
if there is a \LOCAL{} algorithm for approximate inference (within arbitrary total variation error $\delta>0$) with time complexity at most $t(n,\delta)$, 
then $\Problem{M}$ exhibits strong spatial mixing with rate 
\ifabs{$\delta_n(t)=2\min\{\delta\mid t(n,\delta)\le t - 1\}$.

}{\[
\delta_n(t)=2\min\{\delta\mid t(n,\delta)\le t - 1\}. 
\]}
Conversely, for any class of  locally admissible, local Gibbs distributions $\Problem{M}=\{\mu_{(G,\vec{x})}\}$,
if $\Problem{M}$ exhibits strong spatial mixing with rate $\delta_n(t)$, 
then there is a \LOCAL{} algorithm for approximate inference (within arbitrary total variation error $\delta>0$) with time complexity 
\ifabs{$t(n,\delta)=\min\{t\mid \delta_n(t)\le\delta\}+O(1)$.}{\[
t(n,\delta)=\min\{t\mid \delta_n(t)\le\delta\}+O(1).
\]}
\end{theorem}
\ifabs{}{
\begin{proof}
Let $\mu = \mu_{(G, \vec{x})} \in \Problem{M}$ be a joint distribution over $\Sigma^V$, where $G = (V, E)$.   
Let $\sigma, \tau\in\Sigma^\Lambda$ be two feasible configurations on subset $\Lambda\subseteq V$ that differ over $D \subseteq \Lambda$. 
Fix any vertex $v \not\in D$. Suppose that $\dist_G(v, D) = t$. 

Let $\mathcal{A}$ denote the \LOCAL{} algorithm for approximate inference with time complexity $t(n,\delta)$. By Proposition~\ref{prop:deterministic-approx-infer}, we can assume without loss of generality that $\mathcal{A}$ is deterministic. 

For any $\delta>0$, if $t(n,\delta)\le t-1$, then algorithm $\mathcal{A}$ (given the total variation error $\delta$) will return the same marginal distribution $\alg{\mu}_v$ at node $v$ on the two instances $(G,\vec{x}, \sigma)$ and $(G,\vec{x}, \tau)$, because the two instances are indistinguishable for the algorithm at node $v$, and it is guaranteed that $\DTV{\alg{\mu}_v}{\mu^\sigma_v}\le \delta$ and $\DTV{\alg{\mu}_v}{\mu^\tau_v}\le\delta$. Since this holds for any $\delta>0$ such that $t(n,\delta)\le t-1$, we have
\begin{align*}
\DTV{\alg{\mu}_v}{\mu^\sigma_v} &\leq \min\{\delta \mid t(n, \delta) \leq t-1 \},\\
\DTV{\alg{\mu}_v}{\mu^\tau_v} &\leq \min\{\delta \mid t(n, \delta) \leq t-1 \},		
\end{align*}
which  implies 
$\DTV{\mu^\sigma_v}{\mu^\tau_v} \leq  2\min\{\delta \mid t(n, \delta) \leq t-1 \}$
by triangle inequality.
Therefore, $\Problem{M}$ exhibits strong spatial mixing with rate $\delta_n(t)=2\min\{\delta\mid t(n,\delta)\le t - 1\}$. 

\paragraph{}
Conversely, $\Problem{M}$ is a class of locally admissible, local Gibbs distributions and $\Problem{M}$  exhibits strong spatial mixing with rate $\delta_n(t)$. 
We construct a \LOCAL{} algorithm $\mathcal{A}$ for approximate inference with arbitrary total variation error $\delta>0$.

Let $(G, \vec{x}, \tau)$ be an instance, where $G=(V,E)$. The joint distribution $\mu=\mu_{(G, \vec{x})} \in \Problem{M}$ is a Gibbs distribution specified by $(G,\Sigma,\mathcal{F})$, where $q=|\Sigma|$.
Since the Gibbs distribution is local, we assume that there is an $\ell=O(1)$ such that $\forall (f, S) \in \mathcal{F}: \max_{u, v \in S}\dist_G(u, v) \leq \ell$. 
The $\tau\in\Sigma^\Lambda$ is an arbitrary feasible configuration on an arbitrary subset $\Lambda\subseteq V$ and $\mu^\tau$ is the target distribution.

The algorithm $\mathcal{A}$ is described as follows. For each node $v \in V$, node $v$ gathers all information up to distance $t + 2\ell $, where
\[
t=\min\{t'\mid \delta_n(t')\le\delta\},
\]
and simulates the following procedure locally:
\begin{itemize}
\item
Node $v$ extends $\tau$ to a feasible configuration $\tau'$ on subset $\Lambda \cup \Gamma$, where
\begin{align*}
\Gamma  &= B_{t + \ell}(v) \setminus (B_t(v) \cup \Lambda),
\end{align*}
by enumerating all configurations in $\Sigma^\Gamma$. Such feasible $\tau'$ must exist since $\tau$ is feasible.

Since $\tau$ itself is feasible and the distribution $\mu$ is locally admissible, then $\tau'$ is feasible if and only if
\begin{align}
\label{eq-tau'-feasiable}
\prod_{(f,S)\in\mathcal{F}\atop S\subseteq A }f(\tau'_S) >0,	
\end{align}
where $A = B_{t+2\ell}(v) \cap (\Gamma \cup \Lambda)$. Hence, this condition can be checked by $v$ locally. 

\item
Node $v$ returns the marginal distribution $\mu_v^{\tau'}$. Due to the conditional independence guaranteed by Proposition~\ref{prop:cond-ind}, $\mu^{\tau'}_v$ is fully determined by the information in $\Ball_{t+\ell}(v)$.
Specifically, denoted $B=B_{t+\ell}(v)$, and define the set of configurations 
\[
\mathcal{C} = \{\sigma \in \Sigma^B \mid \forall u \in B \cap (\Gamma \cup \Lambda): \sigma_u =\tau'_u \}.
\]
Then for every $c\in \Sigma$, the marginal probability $\mu^{\tau'}_v(c)$ can be computed as
\[
\mu^{\tau'}_v(c)=\frac{\sum_{\sigma\in \mathcal{C} :\sigma_v=c}w_B(\sigma)}{\sum_{\sigma\in \mathcal{C}}w_B(\sigma)},
\]
where $w_B(\sigma)=\prod_{(f,S)\in\mathcal{F}: S\subseteq B}f(\sigma_S)$.
\end{itemize}

Let $\mathcal{S}$ be the set of all feasible configurations $\sigma$ on subset $\Lambda \cup \Gamma$ such that $\sigma$ is consistent with $\tau$ over $\Lambda$. Fix any $\sigma \in \mathcal{S}$. Let $D$ denote the subset on which $\sigma$ and $\tau'$ disagree. Note that $\dist_G(v, D) \geq t$. Due to the strong spatial mixing,
\begin{align*}
\DTV{\mu^{\tau'}_v}{\mu^{\sigma}_v} \leq \delta_n(t) \leq \delta.	
\end{align*}
We couple the two distributions $\mu^{\tau'}_v$ and $\mu^{\tau}_v$ as follows: First sample a $\sigma \in \mathcal{S}$ with probability $\mu^\tau_{\Lambda \cup \Gamma}(\sigma)$, then use the optimal coupling between $\mu^{\tau'}_v$ and $\mu^{\sigma}_v$ to sample a pair $(x, y) \in \Sigma \times \Sigma$. It is easy to verify that the marginal distribution of $y$ resulting from this two-step sampling is just $\mu^{\tau}_v$.
By the coupling Lemma
\begin{align*}
\DTV{\mu^{\tau'}_v}{\mu^{\tau}_v} 
&\leq \Pr[x \neq y]\\
&=\sum_{\sigma \in \mathcal{S}}
\mu^{\tau}_{\Lambda \cup \Gamma}(\sigma) 
\cdot\Pr[x \neq y\mid \sigma]\\
&= \sum_{\sigma \in \mathcal{S}} 
\mu^{\tau}_{\Lambda \cup \Gamma}(\sigma)
\cdot\DTV{\mu^{\tau'}_v}{\mu^{\sigma}_v}	\\
&\leq \delta.
\end{align*}
Hence, $\mathcal{A}$ is a \LOCAL{} algorithm for approximate inference with time complexity at most $t(n,\delta)=t+2\ell=\min\{t\mid \delta_n(t)\le\delta\}+O(1)$.
\end{proof}
}


The strong spatial mixing defined in Definition~\ref{def:ssm} measures the decay of correlation in terms of total variation distance. If we replace the total variation distance $\DTV{\cdot}{\cdot}$ in~\eqref{eq:ssm} with the multiplicative error function $\err{\cdot}{\cdot}$ defined in~\eqref{eq:error-function}, we have an even stronger form of strong spatial mixing, namely the one with decay in multiplicative error.  Several well-known strong spatial mixing results for important classes of Gibbs distributions (e.g.~independent sets, matchings, and graph colorings) were actually established in this stronger form~\cite{weitz2006counting, bayati2007simple, gamarnik2007correlation, gamarnik2013strong}. Here we see this is not a coincidence. Combing Theorem~\ref{thm:ssm-approx-infer} with the boosting lemma (Lemma \ref{lemma-local-boosting}), we have the following corollary.

\begin{corollary}\label{corollary-ssm-boosting}
A class of locally admissible, local Gibbs distributions $\Problem{M}$
exhibits strong spatial mixing with exponential decay at rate $\alpha$ in total variation distance, 
if and only if it exhibits strong spatial mixing with exponential decay at rate $\alpha$ in multiplicative error. 
\end{corollary}
Interestingly, the corollary gives a result in probability theory proved by local computation.

Combining Theorem~\ref{thm:ssm-approx-infer} with the distributed JVV sampler (Theorem~\ref{thm:approx-infer-exact-sample}), we have the followings.

\begin{corollary}\label{coro:ssm-sampling}
For any class of locally admissible, local Gibbs distributions $\Problem{M}=\{\mu_{(G,\vec{x})}\}$, 
if $\Problem{M}$ exhibits strong spatial mixing with exponential decay at rate $\alpha$ for some $\alpha<1$, 
then there is a \LOCAL{} algorithm for exact sampling with time complexity $O(\frac{1}{1-\alpha}\log^3n)$.
\end{corollary}

Combining with the state-of-the-arts for strong spatial mixing in~\cite{bayati2007simple, weitz2006counting, gamarnik2013strong, li2013correlation, song2016counting}, 
the corollary gives us the following \LOCAL{} algorithm for exact sampling:
\begin{itemize}
\item an $O(\sqrt{\Delta}\log^3n)$-round algorithm for sampling matchings in graphs with maximum degree $\Delta$ due to the strong spatial mixing of matchings with exponential decay at rate $1-\Omega({1}/{\sqrt{\Delta}})$~\cite{bayati2007simple};
\item an $O(\log^3n)$-round algorithm for sampling independent sets in graphs with max-degree $\Delta\le 5$, or more generally, for sampling according to the hardcore model (weighted independent sets) with fugacity $\lambda$ up to the uniqueness threshold (where $\lambda<\lambda_c(\Delta)\triangleq{(\Delta-1)^{(\Delta-1)}}/{(\Delta-2)^{\Delta}}$), due to the strong spatial mixing of the model up to the uniqueness threshold~\cite{weitz2006counting};
\item an $O(\log^3n)$-round algorithm for sampling $q$-colorings of triangle-free graphs when $q\ge\alpha\Delta$ for $\alpha>\alpha^*$ where $\alpha^*\approx1.763\ldots$ satisfies $\alpha^*=\exp(\frac{1}{\alpha^*})$, due to the strong spatial mixing proved in~\cite{gamarnik2013strong};
\item an $O(\log^3n)$-round algorithm for sampling according to the anti-ferromagnetic 2-spin model in the interior of the uniqueness regime, due to the strong spatial mixing of the model in the uniqueness regime~\cite{li2013correlation};
\item an $O(\log^3n)$-round algorithm for sampling weighted hypergraph matchings up to the uniqueness threshold (when the weight $\lambda<\lambda_c(r,\Delta)\triangleq\frac{(\Delta-1)^{(\Delta-1)}}{(r-1)(\Delta-2)^{\Delta}}$, where $r$ is the rank of the hypergraph), due to the strong spatial mixing of the model up to the uniqueness threshold~\cite{song2016counting}.
\end{itemize}
The definitions of these models are given in the referred papers.
All these joint distributions are either locally admissible, local Gibbs distributions, or in the case of edge models (e.g.~graph/hypergraph matchings) can be represented as such joint distributions through dualities of graphs/hypergraphs, which preserve the distances. 

For lower bounds, the long-range correlation established in a previous work (\cite{feng2017sampling}, Theorem 5.3) implies an $\Omega(\mathrm{diam})$ lower bound for approximate sampling according to the hardcore model with fugacity $\lambda$ in the non-uniqueness regime (where $\lambda>\lambda_c(\Delta)$). Along with the $O(\log^3n)$ upper bound for exact sampling according to the hardcore model in the uniqueness regime obtained above, we discover for the first time a computational phase transition for local distributed sampling and counting, at the same critical threshold for the computational phase transition discovered for sampling and counting on polynomial-time Turing machines~\cite{weitz2006counting,sly2010computational}.

\ifabs{}{
\section{Conclusion}
We study the complexities of sampling and counting in the \LOCAL{} model, where the counting is represented by a local variant, namely the inference problem. We found that for self-reducible problems, the well known generic relations between sampling and counting on classic polynomial-time Turing machines hold similarly for local computation. Meanwhile, the tractability of these problems by local computation is captured by a decay of correlation property known as the strong spatial mixing. 

Perhaps a lesson we could learn from this research is that it is helpful to model local computation problems as joint distributions, and hence  studying the complexities of these problems is reduced to studying the discrepancies between such problem-specified joint distributions and the distributions that can be generated by local algorithms. This new approach for local computation seems to have much potential.

Several open problems are worth investiageting. First, can we make the \LOCAL{} algorithms in this paper use bounded-size messages and bounded local computation? Such efficient distributed algorithms would necessarily improve the state of the arts of sampling and approximate counting on polynomial-time Turing machines. Second, how should we classify complexities of sampling and counting in local computation, and does there exist a complexity hierarchy? Third, the distributed JVV sampler given in this paper terminates in a fixed number of rounds with bounded locally certifiable failure. Can we make this algorithm Las Vegas, in a sense that the time complexity of the algorithm may be random but once it terminates the algorithm always outputs precisely according to the correct distribution, and still being local? This requires a strategy for non-biased local resampling, which is far from being well understood. So far, it was only discovered for the Lov\'{a}sz-local-lemma-based sampler for restrictive problems under strict conditions~\cite{guo2016uniform}.
}

\bibliographystyle{abbrv}

\end{document}